\documentclass[a4paper]{article}

\usepackage{a4wide}
\usepackage[utf8]{inputenc}
\usepackage{amsfonts}
\usepackage{amsmath}
\usepackage{amsthm}
\usepackage{amssymb} 
\usepackage{enumitem}
\usepackage{bussproofs}
\usepackage{microtype}

\theoremstyle{definition}
\newtheorem{theorem}{Theorem}
\newtheorem{lemma}[theorem]{Lemma}
\newtheorem{proposition}[theorem]{Proposition}

\newtheorem{example}[theorem]{Example}
\newtheorem{definition}[theorem]{Definition}
\newtheorem{remark}[theorem]{Remark}

\newtheorem*{lemma*}{Lemma}

\newcommand{\cA}{\mathcal{A}}
\newcommand{\cB}{\mathcal{B}}

\newcommand{\cI}{\mathcal{I}}
\newcommand{\hI}{\widehat{\cI}}

\newcommand{\hN}{\widehat{N}}

\newcommand{\cP}{\mathcal{P}}
\newcommand{\hP}{\widehat{\cP}}
\newcommand{\cS}{\mathcal{S}}
\newcommand{\cT}{\mathcal{T}}
\newcommand{\cZ}{\mathcal{Z}}

\newcommand{\fa}{\mathfrak{a}}
\newcommand{\fb}{\mathfrak{b}}

\newcommand{\hPhi}{\Phi}

\newcommand{\Nat}{\mathbb{N}}
\newcommand{\Int}{\mathbb{Z}}

\newcommand{\equals}{=}
\newcommand{\<}{\langle}
\renewcommand{\>}{\rangle}

\newcommand{\vars}{\text{\upshape{vars}}}
\newcommand{\consts}{\text{\upshape{consts}}}
\newcommand{\bconsts}{\text{\upshape{bconsts}}}
\newcommand{\fconsts}{\text{\upshape{fconsts}}}

\newcommand{\prop}{\rightrightarrows}
\newcommand{\downcl}{\mathop{\Downarrow}}
\newcommand{\upcl}{\mathop{\Uparrow}}

\newcommand{\subst}[1]{\bigl[#1\bigr]}

\newcommand{\Mapsto}{{\mathop{\mapsto}}}

\newcommand{\False}{\text{False}}

\newcommand{\lvl}{\text{lvl}}

\newcommand{\BsrSli}{\text{BSR(SLI)}}
\newcommand{\BsrSliT}{\text{BSR(SLI$+\cT$)}}
\newcommand{\Tterms}{\mathbb{T}_{\cT}}

\begin{document}
\title{On the Combination of the\\ Bernays--Sch\"onfinkel--Ramsey Fragment\\ with Simple Linear Integer Arithmetic}
\author{
	\begin{tabular}{l}
		Matthias Horbach\\
		\small\textit{Max Planck Institute for Informatics, Saarland Informatics Campus, Saarbr\"ucken, Germany}
	\end{tabular}
	\and
	\begin{tabular}{l}
		Marco Voigt\\
		\small\textit{Max Planck Institute for Informatics, Saarland Informatics Campus, Saarbr\"ucken, Germany,}\\
		\small\textit{Saarbr\"ucken Graduate School of Computer Science}
	\end{tabular}
	\and
	\begin{tabular}{l}
		Christoph Weidenbach \\
		\small\textit{Max Planck Institute for Informatics, Saarland Informatics Campus, Saarbr\"ucken, Germany}
	\end{tabular}
}	
\date{}
\maketitle

\begin{abstract}
	In general, first-order predicate logic extended with linear integer arithmetic is undecidable. 
	We show that the Bernays--Sch\"onfinkel--Ramsey fragment ($\exists^* \forall^*$-sentences) extended with a restricted form of
	linear integer arithmetic is decidable via finite ground instantiation.
	The identified ground instances can be employed to restrict the search space of 
	existing automated reasoning procedures considerably, 
	e.g., when reasoning about quantified properties of array data structures formalized in Bradley, Manna, and Sipma's \emph{array property fragment}.
	Typically, decision procedures for the array property fragment are based on an exhaustive instantiation 
	of universally quantified array indices with all the ground index terms that occur in the formula at hand. 
	Our results reveal that one can get along with significantly fewer instances.
\end{abstract}


\section{Introduction} \label{section:Introduction}

The Bernays-Sch\"onfinkel-Ramsey (BSR) fragment comprises exactly the first-order logic prenex sentences 
with the $\exists^*\forall^*$ quantifier prefix, resulting in a CNF where all occurring function symbols are constants. 
Formulas may contain equality.
Satisfiability of the BSR fragment is decidable and \textsc{NExpTime}-complete~\cite{Lewis1980}.
Its extension with linear arithmetic is undecidable~\cite{Putnam1957,Downey1972,Halpern1991,Fietzke2012}. 

We prove decidability of the restriction to arithmetic constraints of the form $s\triangleleft t$, $x \triangleleft t$, where $\triangleleft$ is one of the standard relations $<, \leq, \equals, \not\equals, \geq, >$ and $s$, $t$ are ground arithmetic terms, and $x\trianglelefteq y$, where $\trianglelefteq$ stands for $\leq$, $\equals$, or $\geq$. 
Underlying the result is the observation that similar to the finite model property of BSR, only finitely many instances of universally quantified clauses with arithmetic constraints need to be considered. 
Our construction is motivated by results from quantifier elimination~\cite{Loos1993} and hierarchic superposition~\cite{Bachmair1994b,Althaus2009,Kruglov2012,Fietzke2012,Baumgartner2013}. 
In particular, the insights gained from the quantifier elimination side lead to instantiation methods that can result in significantly fewer instances than known, more naive approaches for comparable logic fragments generate, such as the original instantiation approach for the \emph{array property fragment}~\cite{Bradley2006,BradleyPhD2007}.
For example, consider the following two clauses ($\wedge$ and $\vee$ bind stronger than $\rightarrow$)\\
	\centerline{
	$\begin{array}{r@{\;}r@{\;\;\rightarrow\;\;}l}
	x_2\neq 5 			&\wedge\; R(x_1)  	& Q(u_1,x_2)\\
	y_1<7 \wedge y_2\leq 2	&	  		& Q(d,y_2) \vee R(y_1)
	\end{array}$
	}
where the variable $u_1$ ranges over a freely selectable domain, $x_i$, $y_i$ are variables over the integers, and the constant
$d$ addresses an element of the same domain that $u_1$ ranges over. All occurring variables are implicitly universally quantified.
Our main result reveals that this clause set is satisfiable if and only if a finite set of ground instances is satisfiable in which
(i) $u_1$ is being instantiated with the constant $d$,
(ii) $x_2$ and $y_2$ are being instantiated with the (abstract) integer values
$5+1$ and $-\infty$, and
(iii) $x_1$ and $y_1$ are being instantiated with $-\infty$ only. 
The instantiation does not need to consider the constraints $y_1<7$, $y_2\leq 2$, because it is sufficient to explore the integers either from
$-\infty$ upwards---in this case upper bounds on integer variables can be ignored---or from $+\infty$ downwards---ignoring lower bounds---, as is similarly done in linear quantifier elimination over the reals~\cite{Loos1993}. 
Moreover, instantiation does not need
to consider the value $5+1$ for $x_1$ and $y_1$, motivated by the fact that the argument $x_1$ of $R$ is not affected by the constraint $x_2 \neq 5$.

The abstract values $-\infty$ and $+\infty$ are represented by Skolem constants over the integers, together with defining axioms. 
For the example, we introduce the fresh Skolem constant $c_{-\infty}$ to represent $-\infty$ (a ``sufficiently small'' value)
together with the axiom $c_{-\infty} < 2$, where $2$ is the smallest occurring constant.
Eventually, we obtain the ground clause set\\
\centerline{
	$\begin{array}{r@{\;}r@{\;\;\rightarrow\;\;}l}
		5+1 \neq 5 					& \wedge\; R(c_{-\infty}) 		& Q(d,5+1)\\
		c_{-\infty}\neq 5 				& \wedge\; R(c_{-\infty}) 		& Q(d,c_{-\infty})\\
		c_{-\infty}<7 \wedge 5+1\leq 2 		& 					& Q(d,5+1) \vee R(c_{-\infty})\\
		c_{-\infty}<7 \wedge c_{-\infty}\leq 2 	& 					& Q(d,c_{-\infty}) \vee R(c_{-\infty})\\
		\multicolumn{3}{c}{c_{-\infty} < 2}
	\end{array}$
}
which has the model $\cA$ with
$c_{-\infty}^{\cA} = 1$, $R^{\cA} = \{1\}$, $Q^{\cA} = \{(d,6), (d,1)\}$.

After developing our instantiation methodology in Section~\ref{section:DecidingBSRwithSLIC},
we show in Sections~\ref{section:StratifiedClauseSets} that our instantiation methods are also compatible with uninterpreted functions and additional background theories under certain syntactic restrictions.
These results are based on an (un)satifiability-preserving embedding of uninterpreted functions into BSR clauses.
There are interesting known logic fragments that fall into this syntactic category: many-sorted clause sets over \emph{stratified vocabularies} \cite{Abadi2010,Korovin2013}, the \emph{array property fragment} \cite{Bradley2006}, and the \emph{finite essentially uninterpreted fragment}, possibly extended with simple integer arithmetic~\cite{Ge2009}.
Consequently, reasoning procedures for these fragments that employ forms of instantiation may benefit from our findings.
The paper ends with a discussion in Section~\ref{section:conclusion}, where we consider the impact of our results on automated reasoning procedures for our and similar logic fragments and outline possible further improvements.

In order to facilitate smooth reading, lengthy technical proofs are only sketched in the main text and presented in full in the appendix. 
The present paper is an extended version of~\cite{VoigtCADE2017}.


\section{Preliminaries}\label{section:preliminaries}

Hierarchic combinations of first-order logic with background theories build upon sorted logic with equality \cite{Bachmair1994b,Baumgartner2013}. We instantiate this framework with the BSR fragment and linear arithmetic over the integers as the \emph{base theory}. The \emph{base sort $\cZ$} shall always be interpreted by the integers $\Int$. For simplicity, we restrict our considerations to a single \emph{free sort $\cS$}, which may be freely interpreted as some nonempty domain, as usual. 

We denote by $V_\cZ$ a countably infinite set of base-sort variables.
\emph{Linear integer arithmetic (LIA) terms} are build from integer constants $0, 1, -1, 2, -2, \ldots$, the operators $+, -$, and the variables from $V_\cZ$.
We moreover allow base-sort constant symbols whose values have to be determined by an interpretation (\emph{Skolem constants}). 
They can be conceived as existentially quantified.
The LIA constraints we consider are of the form $s \triangleleft t$, where $\triangleleft \in \{<, \leq, \equals, \not\equals, \geq, >\}$ and $s$ and $t$ are either LIA 
variables or ground LIA terms. 

In order to hierarchically extend the base theory by the BSR fragment, we introduce the free sort $\cS$, a countably infinite set $V_\cS$ of \emph{free-sort variables}, 
a finite set $\Omega$ of \emph{free (uninterpreted) constant symbols of sort $\cS$} and a finite set $\Pi$ of \emph{free predicate symbols} equipped with sort information. 
Note that every predicate symbol in $\Pi$ has a finite, nonnegative arity and can have a mixed sort over the two sorts $\cZ$ and $\cS$, e.g.\ $P : \cZ \times \cS \times \cZ$.
We use the symbol $\approx$ to denote the built-in equality predicate on $\cS$. 
To avoid confusion, we tacitly assume that no constant or predicate symbol is overloaded, i.e.\ they have a unique sort.

\begin{definition}[BSR\hspace{-0.2ex} with\hspace{-0.2ex} Simple\hspace{-0.2ex} Linear\hspace{-0.2ex} Integer\hspace{-0.2ex} Constraints--\BsrSli]\label{definition:BSRwithConstrSyntax}	
		A \emph{\BsrSli\ clause} has the form $\Lambda \,\|\, \Gamma \to \Delta$, where $\Lambda$, $\Gamma$, $\Delta$ are multisets of atoms satisfying the following conditions.
		\begin{enumerate}[label=(\roman{*}), ref=(\roman{*})]
			\item Every atom in $\Lambda$ is a LIA constraint of the form $s\triangleleft t$ or $x\triangleleft t$ or $x \trianglelefteq y$ 
                              where $s, t$ are ground, $\triangleleft \,{\in} \{<,\leq,\equals,\not\equals,\geq,>\}$,  and $\trianglelefteq {\in} \{\leq,\equals,\geq\}$,
			\item Every atom in $\Gamma$ and $\Delta$ is either an equation $s\approx s'$ with $s,s' \in \Omega \cup V_\cS$, 
                              or a non-equational atom $P(s_1, \ldots, s_m)$, 
				where
				every $s_i$ of sort $\cZ$ must be a variable $x\in V_\cZ$, and
				every $s_i$ of sort $\cS$ may be a variable $u\in V_\cS$ or a constant symbol $c\in\Omega$.
		\end{enumerate}
\end{definition}

We omit the empty multiset left of ``$\to$'' and denote it by $\Box$ right of ``$\to$'' 
(where $\Box$ at the same time stands for \emph{falsity}).
The clause notation separates arithmetic constraints from the \emph{free} (also: \emph{uninterpreted}) part. 
We use the vertical double bar ``$\|$'' to indicate this separation syntactically.  
Intuitively, clauses $\Lambda \,\|\, \Gamma \to \Delta$ can be read as 
$\bigl(\bigwedge\Lambda \wedge \bigwedge\Gamma\bigr) \to \bigvee\Delta$, i.e.\ the multisets $\Lambda, \Gamma$ 
stand for conjunctions of atoms and $\Delta$ stands for a disjunction of atoms. 

Requiring the free part $\Gamma\to\Delta$ of clauses to not contain any base-sort terms apart from variables does not limit expressiveness. 
Every base-sort term $t \not\in V_\cZ$ in the free part can safely be replaced by a fresh base-sort variable $x_t$ when an atomic constraint $x_t \equals t$ is added to the constraint part of the clause (a process known as \emph{purification} or \emph{abstraction}~\cite{Bachmair1994b,Kruglov2012}).

A \emph{hierarchic interpretation} is an algebra $\cA$ which interprets the base sort $\cZ$ as $\cZ^\cA = \Int$, assigns integer values to all occurring base-sort Skolem constants, and interprets all LIA terms and constraints in the standard way. 
Moreover, $\cA$ comprises a nonempty domain $\cS^\cA$, assigns to each free-sort constant symbol $c$ in $\Omega$ a domain element $c^\cA \in \cS^\cA$, and interprets every sorted predicate symbol $P\!:\!\xi_1\times\ldots\times\xi_m$ in $\Pi$ by a set $P^\cA\subseteq \xi_1^\cA\times\ldots\times\xi_m^\cA$, as usual. 

Given a hierarchic interpretation $\cA$ and a sort-respecting \emph{variable assignment} $\beta: V_\cZ\cup V_\cS \to \cZ^\cA \cup \cS^\cA$, we write $\cA(\beta)(s)$ to address the \emph{value of the term $s$ under $\cA$ with respect to the variable assignment $\beta$}. 
The variables occurring in clauses are implicitly universally quantified. 
Therefore, given a clause $C$, we call $\cA$ a \emph{hierarchic model of $C$}, denoted $\cA\models C$, if and only if $\cA,\beta\models C$ holds for every variable assignment $\beta$.
For clause sets $N$, $\cA\models N$ holds if and only if $\cA \models C$ holds true for every clause $C \in N$.
We call a clause $C$ (a clause set $N$) \emph{satisfiable} if and only if there exists a hierarchic model $\cA$ of $C$ (of $N$). 
Two clauses $C,D$ (clause sets $N,M$) are \emph{equisatisfiable} if and only if $C$ ($N$) is satisfiable whenever $D$ ($M$) is satisfiable and vice versa.


Given a \BsrSli\ clause $C$, 
	$\consts(C)$ denotes the set of all constant symbols occurring in $C$. 
	The set $\bconsts(N)$ ($\fconsts(N)$) is the restriction of $\consts(N)$ to base-sort (free-sort) constant symbols. 
	By $\vars(C)$ we denote the set of all variables occurring in $C$.
	Similar notation is used for other syntactic objects.

We define \emph{substitutions} $\sigma$ in the standard way as sort-respecting mappings from variables to terms. 
The \emph{restriction of the domain of a substitution $\sigma$ to a set $V$ of variables} is denoted by $\sigma|_V$ and is defined such that $v\sigma|_V := v\sigma$ for every $v\in V$ and $v\sigma|_V = v$ for every $v\not\in V$. 
While the application of a substitution $\sigma$ to terms, atoms and multisets thereof is defined as usual, we need to be more specific for clauses. 
Consider a \BsrSli\ clause $C := \Lambda \,\|\, \Gamma \to \Delta$ and let $x_1, \ldots, x_k$ denote all base-sort variables occurring in $C$ for which $x_i\sigma \neq x_i$. We then set $C\sigma :=\; \Lambda\sigma, x_1 \equals x_1\sigma,\, \ldots,\, x_k \equals x_k\sigma \,\|\, \Gamma\sigma|_{V_\cS} \rightarrow \Delta\sigma|_{V_\cS}$.

A term, atom, etc.\ is called \emph{ground}, if it does not contain any variables. 
A \BsrSli\ clause $C$ is called \emph{essentially ground} if it does not contain free-sort variables and for every base-sort variable $x$ occurring in $C$ there is a constraint $x\equals t$ in $C$ for some ground LIA term $t$. 
A clause set $N$ is \emph{essentially ground} if all the clauses it contains are essentially ground.

\begin{definition}[Normal Form of \BsrSli\ Clauses]\label{definition:BSRwithConstrNormalform}
A \BsrSli\ clause $\Lambda \,\|\, \Gamma \to \Delta$ is in \emph{normal form} if
	\begin{enumerate}[label=(\arabic{*}), ref=(\arabic{*})]
		\item\label{enum:BSRwithConstrNormalform:I} all non-ground atoms in $\Lambda$ have the form $x \trianglelefteq c$ or $x \leq y$ 
                  (or their symmetric variants) where $c$ is an integer or Skolem constant and $\trianglelefteq {\in} \{\leq,\equals,\geq\}$, 
		\item\label{enum:BSRwithConstrNormalform:IV} all base-sort variables that occur in $\Lambda$ also occur in $\Gamma \to \Delta$, and
		
		\item\label{enum:BSRwithConstrNormalform:V} $\Gamma$ does not contain any equation of the form $u \approx t$. 
	\end{enumerate}
A \BsrSli\ clause set $N$ is in \emph{normal form} if all clauses in $N$ are in normal form and pairwise variable disjoint.
Moreover, we assume that $N$ contains at least one free-sort constant symbol.
\end{definition}

\begin{lemma}
	For every \BsrSli\ clause set $N$ there is an equisatisfiable \BsrSli\ clause set $N'$ in normal form.
\end{lemma}
\begin{proof}[Proof sketch]
	We go through the conditions of Definition~\ref{definition:BSRwithConstrNormalform}.
	\paragraph{Ad~\ref{enum:BSRwithConstrNormalform:I}.}
		Clauses of the form $x \not\equals s, \Lambda' \;\|\; \Gamma \to \Delta$ can be equivalently replaced by two clauses $x < s, \Lambda' \;\|\; \Gamma \to \Delta$ and $x > s, \Lambda' \;\|\; \Gamma \to \Delta$.
		
		Clauses of the form $x \trianglelefteq s, \Lambda' \;\|\; \Gamma \to \Delta$, where $s$ is ground but not a constant symbol and where $\triangleleft {\in} \{ \leq, \equals, \geq \}$, can be replaced---under preservation of (un)satisfiability---by two clauses \mbox{$s \neq c \,\|\, \rightarrow \Box$} and $x \trianglelefteq c, \Lambda' \;\|\; \Gamma \to \Delta$ for some fresh constant symbol $c$.
		
		Similarly, clauses of the form $x > s, \Lambda' \;\|\; \Gamma \to \Delta$ can be replaced---under preservation of (un)satisfiability---by two clauses $s+1 \neq c \;\|\; \to \Box$ and $x \geq c, \Lambda' \;\|\; \Gamma \to \Delta$ for some fresh constant symbol $c$. An analogous solution exists for constraints of the form $x < s$.
					
		Atoms of the form $x \equals y$ can be eliminated by replacing every occurrence of $y$ in the respective clause with $x$---also in the free part of the clause.
			
	\paragraph{Ad~\ref{enum:BSRwithConstrNormalform:IV}.}
		Consider a clause $\Lambda', \Lambda \;\|\; \Gamma \to \Delta$, where every atom in $\Lambda'$ contains a base-sort variable $x$ that does not occur in $\Lambda \;\|\; \Gamma \to \Delta$.
		We remove all atoms $x \not\equals t$ as done above.
		Moreover, we remove all trivial atoms $x \trianglelefteq x$ with $\trianglelefteq {\in} \{ \leq, \equals, \geq \}$ from $\Lambda'$ and partition the result into three parts $\Lambda'_1, \Lambda'_2, \Lambda'_3$ such that
			$\Lambda'_1$ contains exclusively atoms of the form $t < x$ and $t \leq x$,
			$\Lambda'_2$ contains exclusively atoms of the form $x \equals t$,
			$\Lambda'_3$ contains exclusively atoms of the form $x < t$ and $x \leq t$,
		and $t$ stands for some ground base-sort term or some base-sort variable.
		Let $\Lambda''$ be the following set of atoms
			\begin{align*}
				\Lambda'' := \;
							&\Bigl\{ t < t' 		\Bigm| (t < x) \in \Lambda'_1 \text{ and } (x \trianglelefteq t') \in \Lambda'_2 \cup \Lambda'_3 \text{ with } \trianglelefteq {\in} \{\leq, \equals\} \Bigr\} \\
							&\cup \Bigl\{ t < t' 		\Bigm| (t \trianglelefteq x) \in \Lambda'_1 \cup \Lambda'_2 \text{ and } (x < t') \in \Lambda'_3 \text{ with } \trianglelefteq {\in} \{\leq, \equals\} \Bigr\} \\
							&\cup \Bigl\{ t+1 < t' 	\Bigm| (t < x) \in \Lambda'_1 \text{ and } (x < t') \in \Lambda'_3 \Bigr\} \\
							&\cup \Bigl\{ t \leq t' 	\Bigm| (t \leq x) \in \Lambda'_1 \text{ and } (x \trianglelefteq t') \in \Lambda'_2 \cup \Lambda'_3 \text{ with } \trianglelefteq {\in} \{\leq, \equals\} \Bigr\} \\
							&\cup \Bigl\{ t \leq t' 	\Bigm| (x \equals t) \in \Lambda'_2 \text{ and } (x \leq t') \in \Lambda'_3 \Bigr\} \\
							&\cup \Bigl\{ t \equals t'	\Bigm| (x \equals t) \in \Lambda'_2 \text{ and } (x \equals t') \in \Lambda'_2 \Bigr\} ~.
			\end{align*}				
		We replace the clause $\Lambda', \Lambda \;\|\; \Gamma \to \Delta$ by the equivalent one \mbox{$\Lambda'', \Lambda \;\|\; \Gamma \to \Delta$}.

	\paragraph{Ad~\ref{enum:BSRwithConstrNormalform:V}.}
		Clauses of the form $\Lambda \;\|\; u \approx t, \Gamma \to \Delta$ can be equivalently replaced by $(\Lambda \;\|\; \Gamma \to \Delta)\subst{u/t}$, where every occurrence of $u$ is substituted with $t$.
\end{proof}


\section{Instantiation for \BsrSli}\label{section:DecidingBSRwithSLIC}

In this section, we present and prove our main technical result: 
\begin{theorem}\label{theorem:DecidabilityOfBSRwithSimpleBounds}
	Satisfiability of a finite \BsrSli\ clause set $N$ is decidable.
\end{theorem}

In essence, one can show that $N$ is equisatisfiable to a finite set of essentially ground clauses (cf.\ Lemma~\ref{lemma:EquisatisfiabilityBaseInstantiation}). 
There are calculi, such as hierarchic superposition \cite{Bachmair1994b,Althaus2009,Kruglov2012,Fietzke2012,Baumgartner2013} or DPLL(T) \cite{NieuwenhuisEtAl06}, that can decide satisfiability of ground clause sets.
Our decidability result for \BsrSli\ does not come as a surprise, given the similarity to other logic fragments that are known to be decidable, such as the \emph{array property fragment} by Bradley, Manna, and Sipma \cite{Bradley2006,BradleyManna2007} and Ge and de Moura's \emph{finite essentially uninterpreted fragment} extended with simple integer arithmetic constraints \cite{Ge2009}.

More important than the obtained decidability result is the instantiation methodology that we employ, in particular for integer-sort variables.
Typically, decision procedures for the integer-indexed array property fragment are based on an exhaustive instantiation of universally quantified array indices with all the ground index terms that occur in the formula at hand (cf.\ the original approach \cite{Bradley2006,BradleyPhD2007} and standard literature \cite{BradleyManna2007,Kroening2016}). 
In more sophisticated approaches, only a \emph{relevant portion} of the occurring arithmetic terms is singled out before instantiation \cite{Ge2009}.

Our methodology will also be based on a concept of relevant terms, determined by connections between the arguments of predicate symbols 
and instantiation points that are propagated along these connections. 
This part of our method is not specific for the integers but can be applied to the free part of our language as well.
For integer variables, we investigate additional criteria to filter out unnecessary instances, 
inspired by the Loos--Weispfenning quantifier elimination procedure~\cite{Loos1993}.
 We elaborate on this in Sections~\ref{section:BaseSortInstantiation} -- \ref{section:AvoidingNonlinearBlowups}.


\subsection{Instantiation of Integer Variables}\label{section:BaseSortInstantiation}

We first summarize the overall approach for the instantiation of integer variables in an intuitive way. To keep the informal exposition simple, we pretend that all LIA terms are constants from $\Int$. 
We even occasionally refer to the improper values $-\infty$ / $+\infty$ ---``sufficiently small/large'' integers. 
A formal treatment with proper definitions will follow.

Given a finite \BsrSli\ clause set $N$ in normal form, we intend to partition $\Int$ into a set $\cP$ of finitely many subsets $p \in \cP$ such that 
satisfiability of $N$ necessarily leads to the existence of a \emph{uniform} hierarchic model.
\begin{definition}[Uniform Interpretations]
	A hierarchic interpretation $\cA$ is \emph{uniform} with respect to a partition $\cP$ of the integers if and only if for every free predicate symbol $Q$ occurring in $N$, every part $p \in \cP$, and all integers $r_1, r_2 \in p$ we have $\<\ldots, r_1, \ldots\> \in Q^\cA$ if and only if $\<\ldots, r_2, \ldots\> \in Q^\cA$.
\end{definition}	
As soon as we have found such a finite partition $\cP$, we pick one integer value $r_p \in p$ as \emph{representative} from each and every part $p\in\cP$. 
Given a clause $C$ that contains a base-sort variable $x$, and given constant symbols $d_1, \ldots, d_{k}$ whose values cover all these representatives, i.e.\ $\{d_1^\cA, \ldots, d_k^\cA\} = \{r_p \mid p\in\cP\}$, we observe\\
\centerline{$
	\cA\models C \;\text{ if and only if }\; \cA\models \bigl\{C\subst{x/d_i} \bigm| 1\leq i \leq k \bigr\} ~.
$}
This equivalence claims that we can transform universal quantification over the integer domain into finite conjunction over all representatives of subsets in $\cP$. 
Formulated differently, we can extrapolate a model for a universally quantified clause set, if we can find a model of finitely many instances of this clause set.
The formal version of this statement is given in Lemma \ref{lemma:EquisatisfiabilityBaseInstantiation}. 
Uniform hierarchic models play a key role in its proof.

When we extract the partition $\cP$ from the given clause set $N$, we exploit three aspects to increase efficiency:
\begin{enumerate}[label=(E-\roman{*}), ref=(E-\roman{*})]
	\item\label{enum:EfficiencyForInstantiation:I} We group argument positions of free predicate symbols in such a way that the instantiation points relevant for these argument positions are identical.
		This means the variables that are associated to these argument positions, e.g.\ because they occur in such a place in some clause, need to be instantiated only with terms that are relevant for the respective group of argument positions.
		This is illustrated in Example~\ref{example:ConnectedArgumentPositions}.
		
	\item\label{enum:EfficiencyForInstantiation:II} Concerning the \emph{relevant} integer constraints, i.e. the ones that produce instantiation points, one can choose to either stick to lower bounds exclusively, use $-\infty$ as a default (the lowest possible lower bound), and ignore upper bounds.
		Alternatively, one can focus on upper bounds, use $+\infty$ as default, and ignore lower bounds.
		This idea goes back to the Loos--Weispfenning quantifier elimination procedure over the reals \cite{Loos1993}.
		Example~\ref{example:LoosWeispfenningBoundFiltering} gives some intuition.
	
	\item\label{enum:EfficiencyForInstantiation:III} The choice described under \ref{enum:EfficiencyForInstantiation:II} can be made independently for every integer variable  that is to be instantiated.
		See Examples~\ref{example:LoosWeispfenningBoundFiltering} and~\ref{example:EquisatisfiableEssentiallyGroundSet}.
\end{enumerate}

\begin{example}\label{example:ConnectedArgumentPositions}
	Consider the following clauses:\\
	\centerline{$
	\begin{array}{c@{\hspace{4ex}}rclcll}
		C_1 :=	&	1 \leq x_1, x_2 \leq 0		&\| 	&			&\to	&T(x_1), 	&Q(x_1, x_2) ~, \\
		C_2 :=	&	y_3 \leq 7,\; y_1 \leq y_3		&\|	&Q(y_1,y_2)		&\to	&R(y_3) ~, \\
		C_3 :=	&	6 \leq z_1				&\|	&T(z_1)		&\to	&\Box	~.
	\end{array}
	$}	
	The variables $x_1$, $x_2$, $y_1$, $y_2$, $y_3$, and $z_1$ are affected by the constraints in which they occur explicitly. 
	Technically, it is more suitable to speak of the \emph{argument position} $\<T,1\>$ instead of variables $x_1$ and $z_1$ that occur as the first argument of the predicate symbol $T$ in $C_1$ and $C_3$, respectively. 
	Speaking in such terms, argument position $\<T,1\>$ is directly affected by the constraints $1 \leq x_1$ and $6 \leq z_1$, argument position $\<Q,1\>$ is directly affected by $1 \leq x_1$ and $y_1 \leq y_3$, $\<Q,2\>$ is affected by $x_2 \leq 0$, and, finally, $\<R,1\>$ is affected by $y_3 \leq 7$ and $y_1 \leq y_3$.
	Besides such direct effects, there are also indirect effects that have to be taken into account.
	For example, the argument position $\<Q,1\>$ is indirectly affected by the constraint $6 \leq z_1$, because $C_1$ establishes a connection between argument positions $\<T,1\>$ and $\<Q,1\>$ via the simultaneous occurrence of $x_1$ in both argument positions and $\<T,1\>$ is affected by $6 \leq z_1$.
	This is witnessed by the fact that $C_1$ and $C_3$ together logically entail the clause $D := 6 \leq x, y \leq 0 \,\|\, \to Q(x,y)$.
	$D$ can be obtained by a hierarchic superposition step from $C_1$ and $C_3$, for instance.
	Another entailed clause is $6 \leq z, z \leq 7 \,\|\, \rightarrow R(z)$, the (simplified) result of hierarchically resolving $D$ with $C_2$.
	Hence, $\<R,1\>$ is affected by the constraints $6  \leq z$ and $z \leq 7$.
	Speaking in terms of argument positions, this effect can be described as propagation of the lower bound $6 \leq y_1$ from $\<Q,1\>$ to $\<R,1\>$ via the constraint $y_1 \leq y_3$ in $C_2$.
	\qed
\end{example}
One lesson learned from the example is that argument positions can be connected by variable occurrences or constraints of the form $x \leq y$. 
Such links in a clause set $N$ are expressed by the relation $\prop_N$.


\begin{definition}[Connections Between Argument Positions and Argument Position Closures]\label{definition:ConnectedArgumentPositions}
	Let $N$ be a \BsrSli\ clause set in normal form.
	We define $\prop_N$ to be the smallest preorder (i.e.\ a reflexive and transitive relation) over $\Pi\times\Nat$ such that
	$\<Q,j\> \prop_N \<P,i\>$ whenever there is a clause $\Lambda \,\|\, \Gamma \to \Delta$ in $N$ containing free atoms $Q(\ldots,u,\ldots)$ and $P(\ldots,v,\ldots)$ in which the variable $u$ occurs at the $j$-th and the variable $v$ occurs at the $i$-th argument position and 
			\begin{enumerate}[label=(\arabic{*}), ref=(\arabic{*})]	
				\item either $u = v$,
				\item or $u \neq v$, both are of sort $\cZ$ and there are constraints $u \equals v$ or $u \leq v$ in $\Lambda$,
				\item or $u \neq v$, both are of sort $\cS$ and there is an atom $u \approx v$ in $\Gamma$ or in $\Delta$.\footnote{\label{footnote:equationVariables}For any free-sort variable $v$  that occurs in a clause $(\Lambda \,\|\, \Gamma \rightarrow \Delta) \in N$ exclusively in equations, we pretend that $\Delta$ contains an atom $\False_v(v)$, for a fresh predicate symbol $\False_v : \cS$. This is merely a technical assumption. Without it, we would have to treat such variables $v$ as a separate case in all definitions. The atom $\False_v(v)$ is not added ``physically'' to any clause.}			
			\end{enumerate}	
	$\prop_N$ induces downward closed sets $\downcl_N\<P,i\>$ of argument positions, called \emph{argument position closures}:
		${\downcl}_N\<P,i\> := \bigl\{ \<Q,j\> \bigm| \<Q,j\> \prop_N \<P,i\> \bigr\}$.

	Consider a variable $v$ that occurs at the $i$-th argument position of a free atom $P(\ldots, v, \ldots)$ in $N$. We denote the argument position closure related to $v$'s argument position in $N$ by $\downcl_N(v)$, i.e.\ $\downcl_N(v) := \downcl_N\<P,i\>$. 
	If $v$ is a free-sort variable that exclusively occurs in equations, we set $\downcl_N(v) := \downcl\<\False_v,1\>$ (cf. footnote~\textsuperscript{\ref{footnote:equationVariables}}).
	To simplify notation a bit, we write $\prop$, $\downcl \<P,i\>$, and $\downcl(v)$ instead of $\prop_N$, $\downcl_N\<P,i\>$, and $\downcl_N(v)$, when the set $N$ is clear from the context.
\end{definition}
Notice that $\prop$ confined to argument position pairs of the free sort is always symmetric. Asymmetry is only introduced by atomic constraints $x \leq y$.

While the relation $\prop$ indicates how instantiation points are propagated between argument positions, the set $\downcl\<P,i\>$ comprises all argument positions from which  instantiation points are propagated to $\<P,i\>$.
For a variable $v$ the set $\downcl(v)$ contains all argument positions that may produce instantiation points for $v$.

\begin{remark}
	In order to make the propagation relation $\prop$ capture all relevant propagation channels for integer-valued instantiation points, it is vital that the clause set under consideration is in normal form.
	In particular, Condition \ref{enum:BSRwithConstrNormalform:IV} of Definition \ref{definition:BSRwithConstrNormalform} guarantees that every variable $x$ occurring in the constraint part $\Lambda$ of a \BsrSli\ clause $\Lambda \,\|\, \Gamma \rightarrow \Delta$ is associated with an argument position $\<P,i\>$, since $\Gamma$ or $\Delta$ must contain some non-equational atom $P(\ldots, x, \ldots)$.
	
	Moreover, transitivity of $\prop$ entails that two LIA constraints $x \leq y$, $y \leq z$ lead to $\<P,i\> \prop \<Q,j\>$, $\<Q,j\> \prop \<R,k\>$, and $\<P,i\> \prop \<R,k\>$, where $\<P,i\>$, $\<Q,j\>$, and $\<R,k\>$ are intended to be the argument positions associated with $x$, $y$, and $z$, respectively.
	On the other hand, two LIA constraints $x \leq c$, $c \leq y$, where $c$ is a Skolem constant, do not entail propagation of instantiation points from $\<P,i\>$ to $\<Q,j\>$.
	In such cases lower bounds do not have to be propagated for the following reasons. 
	If $y$ is assigned any value smaller than the value assigned to $c$, the constraint $c \leq y$ is violated and, therefore, the clause is satisfied. 
	The constraint $c \leq y$ directly leads to an instantiation point $c$ for $y$, as we shall see in the following definition.
\end{remark}

Next, we collect the instantiation points that are necessary to eliminate base-sort variables by means of finite instantiation. 
\begin{definition}[Instantiation Points for Base-Sort Argument Positions]\label{definition:BaseInstantiationPoints}
	Let $N$ be a \BsrSli\ clause set in normal form and let $P: \xi_1\times\ldots\times \xi_m$ be a free predicate symbol occurring in $N$. For every $i$ with $\xi_i = \cZ$ we define $\cI_{P,i}$ to be the smallest set satisfying the following condition.
	 We have $d\in\cI_{P,i}$ for any constant symbol $d$ for which there exists a clause $C$ in $N$ that contains an atom $P(\ldots, x, \ldots)$ in which $x$ occurs as the $i$-th argument and that contains a constraint $x\equals d$ or $x \geq d$.
\end{definition}

The most apparent peculiarity about this definition is that LIA constraints of the form $x \leq d$ are completely ignored when collecting instantiation points for $x$'s argument position. This is one of the aspects that makes this definition interesting from the efficiency point of view, because the number of instances that we have to consider might decrease considerably in this way. The following example may help to develop an intuitive understanding. 

\begin{example}\label{example:LoosWeispfenningBoundFiltering}
	Consider two clauses $C :=\; 3 \leq x,\, x \leq 5 \;\| \to T(x)$ and $D :=\; x \leq 0 \,\|\, T(x)\to\Box$. 
	Recall that we are looking for a finite partition $\cP$ of $\Int$ such that we can construct a uniform hierarchic model $\cA$ of $\{C, D\}$, i.e.\ for every subset $p\in \cP$ and all integers $r_1, r_2 \in p$ we want $r_1 \in T^\cA$ to hold if and only if $r_2 \in T^\cA$. 
	A natural candidate for $\cP$ is $\{ (-\infty,0], [1,2], [3,5], [6,+\infty) \}$, which takes every LIA constraint in $C$ and $D$ into account. 
	Correspondingly, we find the candidate model $\cA$ with $T^\cA = [3,5]$. Obviously, $\cA$ is uniform with respect to $\cP$.
	
	But there are other interesting possibilities, for instance, the more coarse-grained partition $\{(-\infty, 2], [3, +\infty)\}$ together with the predicate $T^\cA = [3, +\infty)$.
	This latter candidate partition completely ignores the constraints $x \leq 0$ and $x \leq 5$ that constitute upper bounds on $x$ and in this way induces a simpler partition. 
	Dually, we could have concentrated on the upper bounds instead (completely ignoring the lower bounds). 
	This would have led to the partition $\{ (-\infty, 0], [1, 5], [6, +\infty) \}$ and the candidate predicate $T^\cA = [1, 5]$ (or $T^\cA = [1, +\infty)$). 
	Both ways are possible, but the former yields a coarser partition and is thus more attractive, as it will cause fewer instances in the end.
	\qed
\end{example}
The example reveals quite some freedom in choosing an appropriate partition of the integers. A large number of parts directly corresponds to a large number of instantiation points---one for each interval---, and therefore leads to a large number of instances that need to be considered by a reasoning procedure. Hence, regarding efficiency, it is of great importance to keep the partition $\cP$ of $\Int$ coarse.

It remains to address the question of why it is sufficient to consider lower bounds only.
At this point, we content ourselves with an informal explanation.
	Let $\varphi(x)$ be a satisfiable $\wedge$-$\vee$-combination of upper and lower bounds on some integer variable $x$.
	For the sake of simplicity, we assume that every atom in $\varphi$ is of the form $c \leq x$ or $x \leq c$ with $c \in \Int$.
	When we look for some value of $x$ that satisfies $\varphi$, we start from some ``sufficiently small value'' $-\infty$.
	If $-\infty$ yields a solution for $\varphi$, we are done.
	If $[x \Mapsto {-\infty}] \not\models \varphi$, there must be some lower bound in $\varphi$ that prevents $-\infty$ from being a solution.
	In order to find a solution, we successively increase the value of $x$ until a solution is found.
	Interesting test points $r \in \Int$ for $x$ are those where $r-1$ violates some lower bound $c \leq x$ in $\varphi$ and $r$ satisfies the bound, i.e.\ $r = c$.
	Consider two lower bounds $c_1 \leq x$ and $c_2 \leq x$ in $\varphi$ such that $c_1 < c_2$ and $\varphi$ contains no further bound $d \leq x$ with $c_1 < d < c_2$.
	Any assignment $[x \Mapsto r]$ with $c_1 < r < c_2$ satisfies exactly the same lower bounds as the assignment $[x \Mapsto c_1]$ does.
	Moreover, any such assignment satisfies \emph{at most} the upper bounds that $[x \Mapsto c_1]$ satisfies.
	In fact, it may violate some of them.
	Consequently, if neither $[x \Mapsto c_1]$ nor $[x \Mapsto c_2]$ satisfy $\varphi$, then $[x \Mapsto r]$ with $c_1 < r < c_2$ cannot satisfy $\varphi$ either.
	In other words, it suffices to test only values induced by lower bounds.
	The abstract value $-\infty$ serves as the default value, which corresponds to the implicit lower bound $-\infty < x$.

\begin{definition}[Instantiation Points for Base-Sort Argument Position Closures and Induced Partition]\label{definition:BaseVariableInstantiationPoints}
	Let $N$ be a \BsrSli\ clause set in normal form and let $\cA$ be a hierarchic interpretation.
	For every base-sort argument position closure $\downcl \<P,i\>$ induced by $\prop$ we define the following:

	The set $\cI_{\downcl \<P,i\>}$ of \emph{instantiation points for $\downcl \<P,i\>$} is defined by \\
		\centerline{$\cI_{\downcl \<P,i\>} := \{c_{-\infty}\} \cup\, \bigcup_{\<Q,j\> \in \downcl \<P,i\>} \cI_{Q,j}$,}
	where we assume $c_{-\infty}$ to be a distinguished base-sort constant symbol that may occur in $N$.
		
	Let the sequence $r_1, \ldots, r_k$ comprise all integers in the set $\bigl\{c^\cA \bigm| c\in\cI_{\downcl\<P,i\>} \setminus \{c_{-\infty}\} \bigr\}$ ordered so that $r_1 < \ldots < r_k$.
	The partition $\cP_{\downcl\<P,i\>}^\cA$ of the integers into finitely many intervals is defined by\\
		\centerline{$\cP_{\downcl\<P,i\>}^\cA := \bigl\{ (-\infty, r_1-1], [r_1, r_2-1], \ldots, [r_{k-1}, r_k-1], [r_k, +\infty) \bigr\}$.}
\end{definition}
Please note that partitions as described in the definition do always exist, and do not contain empty parts.

\begin{lemma}\label{lemma:BaseVariableInstantiationPointsClosure}
	Let $N$ be a \BsrSli\ clause set in normal form and let $\cA$ be a hierarchic interpretation.
	Consider two argument position pairs $\<Q,j\>, \<P,i\>$ for which $\<Q,j\> \prop \<P,i\>$ holds in $N$.
	Then $\cI_{\downcl\<Q,j\>} \subseteq \cI_{\downcl\<P,i\>}$.
	Moreover, $\cP_{\downcl\<P,i\>}^\cA$ is a refinement of $\cP_{\downcl\<Q,j\>}^\cA$, i.e.\ for every $p \in \cP_{\downcl\<P,i\>}^\cA$ there is some $p' \in \cP_{\downcl\<Q,j\>}^\cA$ such that $p \subseteq p'$.
\end{lemma}

\begin{lemma}\label{lemma:PartitionRepresentatives}
	Let $N$ be a \BsrSli\ clause set in normal form and let $\cA$ be a hierarchic interpretation.
	For every part $p\in \cP_{\downcl\<P,i\>}^\cA$ of the form $p = [r_\ell, r_u]$ or $p = [r_\ell, +\infty)$ we find some constant symbol $c_{\downcl\<P,i\>, p} \in \cI_{\downcl\<P,i\>}$ with $c_{\downcl\<P,i\>, p}^\cA = r_\ell$.
\end{lemma}
Note that the lemma did not say anything about the part $(-\infty, r_u]$ which also belongs to every $\cP_{\downcl\<P,i\>}^\cA$.
Our intention is that the constant symbol $c_{-\infty}$ shall be interpreted by a value from this interval.
Hence, we add the set of clauses $\Psi^{-\infty}_{N} := \bigl\{ (c_{-\infty} \geq c \,\| \rightarrow \Box)  \bigm| c \in \bconsts(N) \setminus \{ c_{-\infty} \} \bigr\}$ whenever necessary.
Note that if $\cA$ is a hierarchic model of a given \BsrSli\ clause set $N$, then $\cA$ can be turned into a model of $\Psi^{-\infty}_{N}$ just by changing the interpretation of $c_{-\infty}$. After this modification $\cA$ is still a model of $N$, if $c_{-\infty}$ does not occur in $N$.

The next lemma shows that we can eliminate base-sort variables $x$ from clauses $C$ in a finite \BsrSli\ clause set $N$ by replacing $C$ with finitely many instances in which $x$ is substituted with the instantiation points that we computed for $x$. 
In addition, the axioms that stipulate the meaning of $c_{-\infty}$ need to be added. 
Iterating this instantiation step for every base-sort variable in $N$ eventually leads to a clause set that is essentially ground with respect to the constraint parts of the clauses it contains (free-sort variables need to be treated separately, of course, see Section~\ref{section:FreeSortInstantiation}).
\begin{lemma}[Finite Integer-Variable Elimination]\label{lemma:EquisatisfiabilityBaseInstantiation}
	Let $N$ be a finite \BsrSli\ clause set in normal form such that, if the constant symbol $c_{-\infty}$ occurs in $N$, then $\Psi^{-\infty}_{N} \subseteq N$.
	Suppose there is a clause $C$ in $N$ which contains a base-sort variable $x$. Let $\hN_x$ be the clause set 
		$\hN_x := \bigl(N\setminus\{ C \}\bigr)  \;\cup\; \bigl\{C\subst{x/c}	\;\bigm|\; c\in\cI_{\downcl_N(x)} \bigr\} \;\cup\; \Psi^{-\infty}_{N}$.
	$N$ is satisfiable if and only if $\hN_x$ is satisfiable. 
\end{lemma}
\begin{proof}[Proof sketch]
	The ``only if''-part is trivial.
			
	The ``if''-part requires a more sophisticated argument.
		In what follows, the notations $\prop$ and $\downcl$ always refer to the original clause set $N$.
		Let $\cA$ be a hierarchic model of $\hN_x$.
		We use $\cA$ to construct the hierarchic model $\cB \models N$ as follows.
		For the domain $\cS^\cB$ we reuse $\cA$'s free domain $\cS^\cA$. 
		For every base-sort or free-sort constant symbol $c\in \consts(N)$ we set $c^{\cB} := c^{\cA}$.
		For every predicate symbol $P:\xi_1\times\ldots\times\xi_m$ that occurs in $N$, for every argument position $i$, $1\leq i\leq m$, with $\xi_i = \cZ$, and for every interval $p\in \cP_{\downcl\<P,i\>}^\cA$ Lemma \ref{lemma:PartitionRepresentatives} and the extra clauses in $\Psi^{-\infty}_{N}$ guarantee the existence of a base-sort constant symbol $c_{\downcl\<P,i\>, p}\in \cI_{\downcl(x)}$, such that $c_{\downcl\<P,i\>, p}^\cA \in p$.
		Based on this observation, we define the family of projection functions $\pi_{\downcl\<P,i\>} : \Int\cup\cS^\cB \to \Int\cup\cS^\cA$ by\\
					\centerline{$
						 \pi_{\downcl\<P,i\>}(\fa) :=
							\begin{cases}
								c_{\downcl\<P,i\>, p}^\cA		&\text{if $\xi_i = \cZ$ and $p\in \cP_{\downcl\<P,i\>}^\cA$}\\
												&\text{is the interval $\fa$ lies in,} \\
								\fa				&\text{if $\xi_i = \cS$.}
							\end{cases}
					$}
		Using the projection functions $\pi_{\downcl\<P,i\>}$, we define the sets $P^\cB$ in such a way that for all domain elements $\fa_1, \ldots, \fa_m$ of appropriate sorts \\
		\centerline{$\bigl\<\fa_1, \ldots, \fa_m\bigr\> \in P^{\cB}$ if and only if $\bigl\<\pi_{\downcl\<P,1\>}(\fa_1), \ldots,$ $\pi_{\downcl\<P,m\>}(\fa_m)\bigr\> \in P^\cA$.}

		We next show $\cB\models N$. Consider any clause $C' := \Lambda' \,\|\, \Gamma' \to \Delta'$ in $N$ and let $\beta : V_\cZ\cup V_\cS\to \Int\cup \cS^\cB$ be some variable assignment. 
		From $\beta$ we derive a special variable assignment $\beta_\pi$ for which we shall infer $\cA,\beta_\pi \models C'$ as an intermediate step:
			$\beta_\pi(v) := \pi_{\downcl(v)}(\beta(v))$
		for every variable $v$.
		If $C' \neq C$, then $\hN_x$ already contains $C'$, and thus $\cA,\beta_\pi\models C'$ must hold.
		In case of $C' = C$, let $p_*$ be the interval in $\cP_{\downcl(x)}^\cA$ containing the value $\beta(x)$, and let $c_*$ be an abbreviation for $c_{\downcl(x), p_*}$. Due to $\beta_\pi(x) = c_*^\cA$ and since $\cA$ is a model of the clause $C\subst{x/c_*}$ in $\hN_x$, we conclude $\cA,\beta_\pi \models C$. 
		Hence, in any case we can deduce $\cA,\beta_\pi\models C'$. 
		By case distinction on why $\cA,\beta_\pi \models C'$ holds, we may use this result to infer $\cB,\beta\models C'$.
		It follows that $\cB\models N$.
\end{proof}

\subsection{Independent Bound Selection}\label{section:IndependentBoundSelection}

By now we have mainly focused on lower bounds as sources for instantiation points.
However, as we have already pointed out (cf.\ \ref{enum:EfficiencyForInstantiation:II} and~\ref{enum:EfficiencyForInstantiation:III} in Section~\ref{section:BaseSortInstantiation} and Example~\ref{example:LoosWeispfenningBoundFiltering}), there is also a dual approach in which upper bounds on integer variables play the central role.
It turns out that the choice between the two approaches can be made independently for every variable that is to be instantiated.
In the interest of efficiency, it makes sense to always choose the approach that results in fewer non-redundant instances or, more abstractly speaking, a set of instances whose satisfiability is easier to decide. 
Example~\ref{example:EquisatisfiableEssentiallyGroundSet} illustrates the overall approach.

Given a clause set $N$ in normal form, the relation $\prop_N$ is defined as before.
Dually to the sets $\downcl_N\<P,i\>$, we define the sets $\upcl_N\<P,i\> := \bigl\{ \<Q,j\> \bigm| \<P,i\> \prop_N \<Q,j\> \bigr\}$, which constitute \emph{upwards closed} sets with respect to $\prop_N$ rather than \emph{downwards closed} sets.
Regarding instantiation points, only LIA constraints $x = d$ and $x \leq d$ lead to $d \in \cI_{\upcl_N(x)}$.
In addition, $c_{+\infty}$ is by default added to every set $\cI_{\upcl_N\<P,i\>}$.
In order to fix the meaning of $c_{+\infty}$, we introduce the set of axioms $\Psi^{+\infty}_{N} := \bigl\{ (c_{+\infty} \leq c \,\| \rightarrow \Box)  \bigm| c \in \bconsts(N) \setminus \{ c_{+\infty} \} \bigr\}$.

The dual versions of Definitions~\ref{definition:BaseInstantiationPoints} and~\ref{definition:BaseVariableInstantiationPoints} and Lemma~\ref{lemma:EquisatisfiabilityBaseInstantiation} read as follows.

\begin{definition}[Dual Instantiation Points for Base-Sort Argument Positions]\label{definition:DualBaseInstantiationPoints}
	Let $N$ be a \linebreak \BsrSli\ clause set in normal form and let $P: \xi_1\times\ldots\times \xi_m$ be a free predicate symbol occurring in $N$. For every $i$ with $\xi_i = \cZ$ we define $\cI^{\text{dual}}_{P,i}$ to be the smallest set satisfying the following condition.
	 We have $d\in\cI^{\text{dual}}_{P,i}$ for any constant symbol $d$ for which there exists a clause $C$ in $N$ that contains an atom $P(\ldots, x, \ldots)$ in which $x$ occurs as the $i$-th argument and that contains a constraint $x\equals d$ or $x \leq d$.
\end{definition}

\begin{definition}[Dual Instantiation Points for Base-Sort Argument Position Closures and Induced Partition]\label{definition:DualBaseVariableInstantiationPoints}
	Let $N$ be a \BsrSli\ clause set in normal form and let $\cA$ be a hierarchic interpretation.
	For every base-sort argument position closure $\upcl \<P,i\>$ induced by $\prop$ we define the following:

	The set $\cI_{\upcl \<P,i\>}$ of \emph{instantiation points for $\upcl \<P,i\>$} is defined by \\
		\centerline{$\cI_{\upcl\<P,i\>} := \{c_{+\infty}\} \cup\, \bigcup_{\<Q,j\> \in \upcl\<P,i\>} \cI^\text{dual}_{Q,j}$.}
	
	Let the sequence $r_1, \ldots, r_k$ comprise all integers in the set $\bigl\{c^\cA \bigm| c\in\cI_{\upcl\<P,i\>} \setminus \{ c_{+\infty}\} \bigr\}$ ordered so that $r_1 < \ldots < r_k$.
	The partition $\cP_{\upcl\<P,i\>}^\cA$ of the integers into finitely many intervals is defined by\\
		\centerline{$\cP_{\upcl\<P,i\>}^\cA := \bigl\{ (-\infty , r_1], [r_1+1, r_2], \ldots, [r_{k-1}+1, r_k], [r_k+1, +\infty) \bigr\}$.}
\end{definition}

In the following lemma we refer to the set\\
\centerline{$\Psi^{+\infty}_{N} := \bigl\{ (c_{+\infty} \leq c \,\| \rightarrow \Box)  \bigm| c \in \bconsts(N) \setminus \{ c_{+\infty} \} \bigr\}$.}
\begin{lemma}\label{lemma:EquisatisfiabilityDualBaseInstantiation}
	Let $N$ be a finite \BsrSli\ clause set in normal form such that, if the constant symbol $c_{+\infty}$ occurs in $N$, then $\Psi^{+\infty}_{N} \subseteq N$.
	Suppose there is a clause $C$ in $N$ which contains a base-sort variable $x$. Let
		$\hN_x := \bigl(N\setminus\{ C \}\bigr)  \;\cup\; \bigl\{C\subst{x/c}	\;\bigm|\; c\in\cI_{\upcl_N(x)} \bigr\} \;\cup\; \Psi^{+\infty}_{N}$.		
	$N$ is satisfiable if and only if $\hN_x$ is satisfiable. 
\end{lemma}

In both, Lemma~\ref{lemma:EquisatisfiabilityBaseInstantiation} and its dual version, Lemma~\ref{lemma:EquisatisfiabilityDualBaseInstantiation}, the equisatisfiable instantiation can be applied to the respective variable independently of the instantiation steps that have already been done or are still to be done in the future.
This means, we can choose independently, whether to stick to the lower or upper bounds for instantiation.
This choice can, for example, be made depending on the number of non-redundant instances that have to be generated.

\begin{example}\label{example:EquisatisfiableEssentiallyGroundSet}
	Consider the following \BsrSli\ clause set $N$:\\
	\small
	\centerline{$
		\begin{array}{rclcl@{\hspace{0.5ex}}l}
			1 \leq x_1, x_2 \leq 0		&\| 	&			&\to	&T(x_1), 	&Q(x_1, x_2) \\
			y_3 \leq 7,\; y_1 \leq y_3		&\|	&Q(y_1,y_2)		&\to	&R(y_3)  \\
			6 \leq z_1,\; z_1 \leq 9		&\|	&T(z_1)		&\to	&\Box	
		\end{array}
	$}
	\normalsize
	We intend to instantiate the variables $y_3, y_1, x_1, z_1$ in this order.
	For $y_3$ we can choose between $\cI_{\downcl_N(y_3)} = \{ c_{-\infty}, 1, 6 \}$ and $\cI_{\upcl_N(y_3)} = \{ 7, c_{+\infty} \}$.
	Using the latter option, we obtain the instances\\
	\small
	\centerline{$
	\begin{array}{@{\hspace{-1ex}}rclcl@{\hspace{0.5ex}}l}
		7 \leq 7, y_1 \leq 7, y_3 = 7						&\|	&Q(y_1,y_2)		&\to	&R(y_3)  \\
		c_{+\infty} \leq 7, y_1 \leq c_{+\infty}, y_3 = c_{+\infty}		&\|	&Q(y_1,y_2)		&\to	&R(y_3) 
	\end{array}
	$}
	\normalsize
	plus the clauses in $\Psi^{+\infty}_{N}$.
	The constraint $7 \leq 7$ can be removed, as it is redundant.
	The second instance can be dropped immediately, since the constraint $c_{+\infty} \leq 7$ is false in any model satisfying $\Psi^{+\infty}_{N}$. 
	Dual simplifications can be applied to constraints with $c_{-\infty}$.
	Let $N'$ contain the clauses in $\Psi^{+\infty}_{N}$ and the clauses\\
	\small
	\centerline{$
		\begin{array}{rclcl@{\hspace{0.5ex}}l}
			1 \leq x_1, x_2 \leq 0		&\| 	&			&\to	&T(x_1), 	&Q(x_1, x_2)  \\
			y_1 \leq 7, y_3 = 7			&\|	&Q(y_1,y_2)		&\to	&R(y_3)  \\
			6 \leq z_1, z_1 \leq 9		&\|	&T(z_1)		&\to	&\Box	
		\end{array}
	$}
	\normalsize
	For $y_1$ we use $\cI_{\downcl_{N'}(y_1)} = \{ c_{-\infty}, 1, 6 \}$ rather than $\cI_{\upcl_{N'}(y_1)} = \{ 7, 9, c_{+\infty} \}$ for instantiation and obtain $N''$ (after simplification):\\
	\small
	\centerline{$
		\begin{array}{rclcl@{\hspace{0.5ex}}l}
			1 \leq x_1, x_2 \leq 0		&\| 	&			&\to	&T(x_1), 	&Q(x_1, x_2) \\
			y_3 = 7, y_1 = c_{-\infty}		&\|	&Q(y_1,y_2)		&\to	&R(y_3)  \\
			y_3 = 7, y_1 = 1			&\|	&Q(y_1,y_2)		&\to	&R(y_3) \\
			y_3 = 7, y_1 = 6			&\|	&Q(y_1,y_2)		&\to	&R(y_3) \\
			6 \leq z_1,\; z_1 \leq 9		&\|	&T(z_1)		&\to	&\Box	
		\end{array}
	$}
	\normalsize
	plus the clauses in $\Psi^{-\infty}_{N}$ and $\Psi^{+\infty}_{N}$ and plus the clause $c_{-\infty} \geq c_{+\infty} \| \rightarrow \Box$.
	The sets of instantiation points for $x_1$ in $N''$ are $\cI_{\downcl_{N''}(x_1)} = \{c_{-\infty}, 1, 6\}$ and $\cI_{\upcl_{N''}(x_1)} = \{c_{-\infty},1,6,9, c_{+\infty}\}$.
	The latter set nicely illustrates how instantiation sets for particular variables can evolve during the incremental process of instantiation.
	We take the set with fewer instantiation points and obtain $N'''$:\\
	\small
	\centerline{$
		\begin{array}{rclcl@{\hspace{0.5ex}}l}
			x_2 \leq 0, x_1 = 1			&\| 	&			&\to	&T(x_1), 	&Q(x_1, x_2)  \\
			x_2 \leq 0, x_1 = 6			&\| 	&			&\to	&T(x_1), 	&Q(x_1, x_2)  \\
			y_3 = 7, y_1 = c_{-\infty}		&\|	&Q(y_1,y_2)		&\to	&R(y_3)  \\
			y_3 = 7, y_1 = 1			&\|	&Q(y_1,y_2)		&\to	&R(y_3)  \\
			y_3 = 7, y_1 = 6			&\|	&Q(y_1,y_2)		&\to	&R(y_3)  \\
			6 \leq z_1,\; z_1 \leq 9		&\|	&T(z_1)		&\to	&\Box	
		\end{array}
	$}
	\normalsize
	plus $\Psi^{-\infty}_{N} \cup \Psi^{+\infty}_{N} \cup \{c_{-\infty} \geq c_{+\infty} \| \rightarrow \Box\}$.
	We instantiate $z_1$ using the set $\cI_{\downcl_{N'''}(z_1)} = \{c_{-\infty}, 1, 6\}$ and not $\cI_{\upcl_{N'''}(z_1)} = \{c_{-\infty}, 1, 6, 9, c_{+\infty}\}$:\\
	\small
	\centerline{$
		\begin{array}{rclcl@{\hspace{0.5ex}}l}
			x_2 \leq 0, x_1 = 1			&\| 	&			&\to	&T(x_1), 	&Q(x_1, x_2) \\
			x_2 \leq 0, x_1 = 6			&\| 	&			&\to	&T(x_1), 	&Q(x_1, x_2) \\
			y_3 = 7, y_1 = c_{-\infty}		&\|	&Q(y_1,y_2)		&\to	&R(y_3) \\
			y_3 = 7, y_1 = 1			&\|	&Q(y_1,y_2)		&\to	&R(y_3) \\
			y_3 = 7, y_1 = 6			&\|	&Q(y_1,y_2)		&\to	&R(y_3) \\
			z_1 = 6				&\|	&T(z_1)		&\to	&\Box	
		\end{array}
	$}
	\normalsize
	\noindent
	plus $\Psi^{-\infty}_{N} \cup \Psi^{+\infty}_{N} \cup \{c_{-\infty} \geq c_{+\infty} \| \rightarrow \Box\}$.
	Until now, we have introduced $6$ non-redundant instances.	
	A completely naive instantiation approach where $x_1, y_1, y_3, z_1$ are instantiated with all occurring constant symbols $0, 1, 6, 7, 9$ leads to $17$ non-redundant instances.
	This corresponds to the originally proposed method for the \emph{array property fragment}, cf.~\cite{Bradley2006}. 
	A more sophisticated instantiation approach where $x_1, y_1, y_3, z_1$ are instantiated with $1, 6, 7, 9$ (as there is no connection from $0$ to $x_1$, $y_1$, $y_3$, $z_1$) leads to $13$ non-redundant instances. 
	For instance, the methods described in~\cite{Ge2009} produce this set of instances.\\[1ex]
	\centerline{
		\begin{tabular}{llc}
			instantiation				& instantiation points			& non-redundant \\
			method 	 	  		& for $y_3$, $y_1$, $x_1$, $z_1$	& instances \\[0.5ex]
			\hline
			exhaustive \cite{Bradley2006}	& 4 times $\{0,1,3,6,9\}$			& 17 \\[0.5ex]
			\hspace{-2ex}
			\begin{tabular}{ll}
				filtered by argu- \\
				ment positions \cite{Ge2009}
			\end{tabular}
								& 4 times $\{1,6,7,9\}$			& 13 \\[0.5ex]
			our approach			& \hspace{-1.5ex}\begin{tabular}{l@{\hspace{0.5ex}}l} $\{ 7, c_{+\infty} \}$, &$\{ c_{-\infty}, 1, 6 \}$,\\ $\{c_{-\infty}, 1, 6\}$, &$\{c_{-\infty}, 1, 6\}$	\end{tabular}					& 6
		\end{tabular}
	}
	\qed
\end{example}
The example shows that our approach to instantiation can reduce the number of introduced instances substantially.
Our approach is particularly beneficial in cases where argument positions are to a large degree independent (i.e.\ not connected via $\prop$) and/or where there is a strong imbalance between the number of upper and lower bounds that are connected to a certain argument position.
To illustrate the latter, consider a clause $C$ in a \BsrSli\ clause set $N$ with base-sort variables $x_1, \ldots, x_n$, which are all pairwise connected via $\prop$, and which are subject (directly or via $\prop$) to $\ell$ lower bounds $c_1 \leq z_1, \ldots, c_\ell \leq z_\ell$ and $k$ upper bounds $z'_1 \leq d_1, \ldots, z'_k \leq d_k$.
Assume that the $c_1, \ldots, c_\ell, d_1, \ldots, d_k$ are all pairwise distinct and different from $c_{-\infty}$.
Moreover, suppose $\ell < k$.
Instantiating the variables $x_1, \ldots, x_n$ in $C$ with all constant symbols $c_1, \ldots, c_\ell, d_1, \ldots, d_k$ yields $(\ell+k)^n$ instances. 
In constrast, by Lemma~\ref{lemma:EquisatisfiabilityDualBaseInstantiation}, it is sufficient to consider the instances of $C$ resulting from instantiating every $x_i$ with $c_{-\infty}, c_1, \ldots, c_\ell$.
Hence, only $(\ell+1)^n$ instances need to be considered.
In the extreme case where $\ell = 0$  and $k > 0$, our approach only needs a single instance instead of $k^n$ instances.

\subsection{Instantiation of Free-Sort Variables}\label{section:FreeSortInstantiation}

We can also follow an instantiation approach for free-sort variables.
In a nutshell, we collect only \emph{relevant} instantiation points for a given argument position (cf.\ \ref{enum:EfficiencyForInstantiation:I}).  
A similar approach is taken in~\cite{Ge2009}.

\begin{definition}[Instantiation Points for Free-Sort Argument Positions]\label{definition:FreeInstantiationPoints}
	Let $N$ be a \BsrSli\ clause set in normal form and let $P: \xi_1\times\ldots\times \xi_m$ be a free predicate symbol occurring in $N$ (we pretend that $P$ also reaches over the predicate symbols $\False_v:\cS$, cf.\ footnote~\textsuperscript{\ref{footnote:equationVariables}}). 
	For every $i$ with $\xi_i = \cS$ we define $\cI_{P,i}$ to be the smallest set satisfying the following conditions:
	\begin{enumerate}[label=(\alph{*}), ref=(\alph{*})]
		\item\label{enum:FreeInstantiationPoints:I} $d\in\cI_{P,i}$ for any constant symbol $d$ for which there exists an atom $P(\ldots, d, \ldots)$ in $N$ with $d$ in the $i$-th argument position,
		\item\label{enum:FreeInstantiationPoints:II} $\cI_{P,i} = \fconsts(N)$ for any clause $\Lambda\|\Gamma\rightarrow\Delta$ in $N$ such that $\Gamma\rightarrow \Delta$ contains $P(\ldots, u, \ldots)$ in which $u$ occurs as the $i$-th argument and $\Delta$ contains an atom of the form $u \approx t$ where $t$ is either a variable or a constant symbol.
	\end{enumerate}
\end{definition}

\begin{definition}[Instantiation Points for Free-Sort Argument Position Closures]\label{definition:FreeVariableInstantiationPoints}
	Let $N$ be a \BsrSli\ clause set in normal form.
	For every free-sort argument position closure $\downcl \<P,i\>$ induced by $\prop$ we define the set $\cI_{\downcl \<P,i\>}$ of \emph{instantiation points for $\downcl \<P,i\>$} by $\cI_{\downcl \<P,i\>} := \bigcup_{\<Q,j\> \in \downcl \<P,i\>} \cI_{Q,j}$, if this results in a non-empty set. 
	Otherwise, we set $\cI_{\downcl \<P,i\>} := \{d\}$ for an arbitrarily chosen $d \in \fconsts(N)$.
\end{definition}

\begin{lemma}\label{lemma:EquisatisfiabilityFreeInstantiation}
	Let $N$ be a finite \BsrSli\ clause set in normal form.
	Suppose there is a clause $C$ in $N$ which contains a free-sort variable $u$. Let $\hN_u := \bigl(N\setminus\{ C \}\bigr) \;\cup\; \bigl\{C\subst{u/c} \;\bigm|\; c\in\cI_{\downcl_N(u)} \bigr\}$.
	$N$ is satisfiable if and only if $\hN_u$ is satisfiable. 
\end{lemma}
\begin{proof}[Proof sketch]
	The proof of the ``if''-part proceeds along similar lines as in the proof of Lemma~\ref{lemma:EquisatisfiabilityBaseInstantiation}.
	The main difference is the family of projection functions $\pi_{\downcl\<P,i\>} : \Int\cup\cS^\cB \to \Int\cup\cS^\cA$, which we now define by\\
			\centerline{$
				\pi_{\downcl\<P,i\>}(\fa) :=
					\begin{cases}
						\fa					&\text{if $\xi_i = \cS$ and $\fa = c^\cA$ for some $c \in \cI_{\downcl\<P,i\>}$,} \\
						d_{\downcl\<P,i\>}^\cA		&\text{if $\xi_i = \cS$ and $\fa \neq c^\cA$ for every $c \in \cI_{\downcl\<P,i\>}$,} \\
						\fa					&\text{if $\xi_i = \cZ$,}
					\end{cases}
			$}
	where for every argument position closure $\downcl\<P,i\>$ we fix some \emph{default instantiation point} $d_{\downcl\<P,i\>} \in \cI_{\downcl\<P,i\>}$, for which we choose an arbitrary constant symbol from $\cI_{\downcl\<P,i\>}$.
\end{proof}

\subsection{Avoiding Immediate Blowups}\label{section:AvoidingNonlinearBlowups}

Compared to naive approaches to instantiation of integer-sort and free-sort variables, our methods produce exponentially fewer instances in certain cases.
Still, the number of instances can become very large.
Consider again the clause $C :=\; y_3 \leq 7, y_1 \leq y_3 \,\|\, Q(y_1,y_2) \rightarrow R(y_3)$ from Example~\ref{example:EquisatisfiableEssentiallyGroundSet}.
Instantiating $y_3$ with $\cI_{\upcl(y_3)} = \{7, c_{+\infty}\}$ first and then $y_1$ with $\cI_{\downcl(y_1)} = \{c_{-\infty}, 1, 6\}$ leads to $|\cI_{\upcl(y_3)}| \cdot |\cI_{\downcl(y_1)}| = 6$ instances of $C$ (before simplification):\\
	\small
	\centerline{$
		\begin{array}{rclcll}
			7 \leq 7,\; c_{-\infty} \leq 7, y_3 = 7, y_1 = c_{-\infty}						&\|	&Q(y_1,y_2)		&\to	&R(y_3) ~, \\
			7 \leq 7,\; 1 \leq 7, y_3 = 7, y_1 = 1									&\|	&Q(y_1,y_2)		&\to	&R(y_3) ~, \\
			7 \leq 7,\; 6 \leq 7, y_3 = 7, y_1 = 6									&\|	&Q(y_1,y_2)		&\to	&R(y_3) ~, \\
			c_{+\infty} \leq 7,\; c_{-\infty} \leq c_{+\infty}, y_3 = c_{+\infty}, y_1 = c_{-\infty}		&\|	&Q(y_1,y_2)		&\to	&R(y_3) ~, \\
			c_{+\infty} \leq 7,\; 1 \leq c_{+\infty}, y_3 = c_{+\infty}, y_1 = 1				&\|	&Q(y_1,y_2)		&\to	&R(y_3) ~, \\
			c_{+\infty} \leq 7,\; 6 \leq c_{+\infty}, y_3 = c_{+\infty}, y_1 = 6				&\|	&Q(y_1,y_2)		&\to	&R(y_3) ~.
		\end{array}
	$}
	\normalsize
We refer to this set as $M_1$.
Although simplification will remove the last three clauses, as they are redundant, we add instances to the clause set without knowing whether they are really necessary for showing unsatisfiability, for instance.

We can, on the other hand, leave it to the theorem prover to decide when instantiation is appropriate.
In order to do so, we need to encode the information contained in the computed sets of instantiation points into the clause set using a standard technique.
Regarding the above example, this leads to the set $M_2$ containing $|\cI_{\upcl(y_3)}| + |\cI_{\downcl(y_1)}| + 1 = 6$ clauses:\\
	\small
	\centerline{$
		\begin{array}{rclcll}
			y_3 \leq 7,\; y_1 \leq y_3		&\|	&S_{y_3}(y_3), S_{y_1}(y_1), Q(y_1,y_2)		&\rightarrow	&R(y_3) ~, \\
			y'_3 = 7				&\|	&	\qquad\rightarrow	S_{y_3}(y'_3) ~, \\
			y''_3 = c_{+\infty}			&\|	&	\qquad\rightarrow	S_{y_3}(y''_3) ~, \\
			y'_1 = c_{-\infty}			&\|	&	\qquad\rightarrow	S_{y_1}(y'_1) ~, \\
			y''_1 = 1				&\|	&	\qquad\rightarrow	S_{y_1}(y''_1) ~, \\
			y'''_1 = 6				&\|	&	\qquad\rightarrow	S_{y_1}(y'''_1) ~.
		\end{array}
	$}
	\normalsize
Hierarchic superposition, for instance, can generate the clauses in $M_1$ from the clauses in $M_2$ by resolving over the atoms $S_{y_1}(\ldots)$.
However, in order to derive the empty clause from an unsatisfiable clause set, it is not always necessary to generate all instances.
Instead, a refuting theorem prover can use the information encoded in $M_2$ to instantiate $C$ on demand.
This might prevent a non-linear blowup caused by immediate instantiation with all instantiation points, since we trade the multiplication in $|M_1| = |\cI_{\upcl(y_3)}| \cdot |\cI_{\downcl(y_1)}|$ for addition in $|M_2| = |\cI_{\upcl(y_3)}| + |\cI_{\downcl(y_1)}| + 1$.


\section{Stratified Clause Sets}\label{section:StratifiedClauseSets}

In this section we treat certain clause sets with uninterpreted non-constant function symbols.
By a transformation into an equisatisfiable set of BSR clauses, we show that our instantiation methods are also applicable in such settings.

\begin{definition}\label{definition:StratifiedFunctionalClauseSet}
	Let $N$ be a finite set of variable-disjoint first-order clauses in which also non-constant function symbols occur.
	By $\Pi_N$ and $\Omega_N$ we denote the set of occurring predicate symbols and function symbols (including constants), respectively.
	$N$ is considered to be \emph{stratified} if we can define a mapping $\lvl_N : (\Pi_N \cup \Omega_N) \times \Nat \to \Nat$ that maps argument position pairs (of predicate and function symbols) to nonnegative integers such that the following conditions are satisfied.
	\begin{enumerate}[label=(\alph{*}), ref=(\alph{*})]
		\item For every function symbol $f : \xi_1 \times \ldots \times \xi_m \to \xi_{m+1}$ and every $i \leq m$ we have $\lvl_N\<f,i\> > \lvl_N\<f,m+1\>$.
		\item For every (sub)term $g(s_1, \ldots, s_{k-1}, f(t_1, \ldots, t_m), s_{k+1}, \ldots, s_{m'})$ occurring in $N$ we have \linebreak $\lvl_N\<f, m+1\> = \lvl_N\<g,k\>$.
			This includes the case where $f$ is a constant symbol and $m = 0$.
			Moreover, this also includes the case where $g$ is replaced with a predicate symbol $P$.
		\item For every variable $v$ that occurs in two (sub)terms $f(s_1, \ldots, s_{k-1}, v, s_{k+1}, \ldots,$ $s_m)$ and $g(t_1, \ldots, t_{k'-1}, v, t_{k'+1}, \ldots, t_{m'})$ in $N$ we have $\lvl_N\hspace{-0.2ex}\<f,k\> \!=\! \lvl_N\hspace{-0.2ex}\<g,k'\>$.
			The same applies, if $f$ or $g$ or both are replaced with predicate symbols.
		\item For every equation $f(s_1, \ldots, s_m) \approx g(t_1, \ldots, t_{m'})$ we have $\lvl_N\<f,m+1\> = \lvl_N\<g,m'+1\>$.
			This includes the cases where $f$ or $g$ or both are constant symbols (with $m = 0$ or $m' = 0$ or both, respectively). 	
	\end{enumerate}
\end{definition}
Several known logic fragments fall into this syntactic category:
many-sorted clauses over \emph{stratified vocabularies} as described in \cite{Abadi2010,Korovin2013}, and
clauses belonging to the \emph{finite essentially uninterpreted fragment} (cf.\ Proposition~2 in \cite{Ge2009}).

\begin{lemma}\label{lemma:FunctionFlattening}
	Let $C = \Gamma \rightarrow \Delta$ be a first-order clause and 
	let $f_1, \ldots, f_n$ be a list of all uninterpreted non-constant function symbols occurring in $C$.
	Let $R_1, \ldots, R_n$ be distinct predicate symbols that do not occur in $C$ and that have the
		 sort $R_i : \xi_1 \times \ldots \times \xi_m \times \xi_{m+1}$, if and only if $f_i$ has the sort $\xi_1 \times \ldots \times \xi_m \to \xi_{m+1}$.
	Let $\Phi_1$ and $\Phi_2$ be the following sets of sentences: \\
		\centerline{$\Phi_1 := \bigl\{ \forall x_1 \ldots x_m u v.\; R_i(x_1, \ldots, x_m, u) \wedge R_i(x_1, \ldots, x_m, v) \rightarrow u \approx v  \bigm| 1 \leq i \leq n \bigr\}$}
	and
			$\Phi_2 := \bigl\{ \forall x_1 \ldots x_m \exists v.\; R_i(x_1, \ldots, x_m, v) \bigm| 1 \leq i \leq n \bigr\}$.
	There is a clause $D$ that does not contain non-constant function symbols and for which the set $\{D\} \cup \Phi_1 \cup \Phi_2$ is equisatisfiable to $C$.
\end{lemma}
\begin{proof}[Proof sketch]
	We apply the following flattening rules.
	$v$ stands for a fresh variable that has not occurred yet.
	$P$ ranges over predicate symbols different from~$\approx$.
	$\bar{s}$ and $\bar{t}$ stand for tuples of arguments.\\[-3ex]
	\begin{tabular}{p{42ex} p{42ex}}
		\begin{prooftree}
			\AxiomC{$\Gamma, f_i(\bar{s}) \approx f_j(\bar{t}\,) \rightarrow \Delta$}
			\RightLabel{\scriptsize (fun-fun \!left) \normalsize}
			\UnaryInfC{$\Gamma, R_i(\bar{s},v), R_j(\bar{t},v) \rightarrow \Delta$}
		\end{prooftree}
		&
		\begin{prooftree}
			\AxiomC{$\Gamma \rightarrow \Delta, f_i(\bar{s}) \approx f_j(\bar{t}\,)$}
			\RightLabel{\scriptsize (fun-fun \!right) \normalsize}
			\UnaryInfC{$\Gamma, R_i(\bar{s},v) \rightarrow \Delta, R_j(\bar{t},v)$}
		\end{prooftree}
		\\[-3ex]
		\begin{prooftree}
			\AxiomC{$\Gamma, f_i(\bar{s}) \approx c \rightarrow \Delta$}
			\RightLabel{\scriptsize (fun-const \!left) \normalsize}
			\UnaryInfC{$\Gamma, R_i(\bar{s},c) \rightarrow \Delta$}
		\end{prooftree}
		&
		\begin{prooftree}
			\AxiomC{$\Gamma \rightarrow \Delta, f_i(\bar{s}) \approx c$}
			\RightLabel{\scriptsize (fun-const \!right) \normalsize}
			\UnaryInfC{$\Gamma \rightarrow \Delta, R_i(\bar{s},c)$}
		\end{prooftree}
		\\[-3ex]
		\begin{prooftree}
			\AxiomC{$\Gamma, f_i(\bar{s}) \approx x \rightarrow \Delta$}
			\RightLabel{\scriptsize (fun-var \!left) \normalsize}
			\UnaryInfC{$\Gamma, R_i(\bar{s},x) \rightarrow \Delta$}
		\end{prooftree}
		&
		\begin{prooftree}
			\AxiomC{$\Gamma \rightarrow \Delta, f_i(\bar{s}) \approx x$}
			\RightLabel{\scriptsize (fun-var \!right) \normalsize}
			\UnaryInfC{$\Gamma \rightarrow \Delta, R_i(\bar{s},x)$}
		\end{prooftree}
		\\[-3ex]
		\begin{prooftree}
			\AxiomC{$\Gamma, P(\ldots, f_i(\bar{s}), \ldots) \rightarrow \Delta$}
			\RightLabel{\scriptsize (fun \!left) \normalsize}
			\UnaryInfC{$\Gamma, R_i(\bar{s},v), P(\ldots, v, \ldots) \rightarrow \Delta$}
		\end{prooftree}
		&
		\begin{prooftree}
			\AxiomC{$\Gamma \rightarrow \Delta, P(\ldots, f_i(\bar{s}), \ldots)$}
			\RightLabel{\scriptsize (fun \!right) \normalsize}
			\UnaryInfC{$\Gamma, \hspace{-0.2ex}R_i(\bar{s},v) \rightarrow \Delta, \hspace{-0.2ex}P(\ldots, v, \ldots)$}
		\end{prooftree}
	\end{tabular}\\
\end{proof}
Given a BSR clause $\Gamma \rightarrow \Delta$, we consider an atom $R_j(\bar{t},v)$ in $\Delta$ to be \emph{guarded}, if there is also an atom $R_i(\bar{s}, v)$ in $\Gamma$.
With the exception of the rule \texttt{(fun-var right)} the flattening rules presented in the proof of Lemma~\ref{lemma:FunctionFlattening} preserve guardedness of atoms in $\Delta$ and introduce atoms $R_j(\bar{t}, v)$ on the right-hand side of a clause only if at the same time a corresponding guard is introduced on the left-hand side of the clause.

Hence, if we are given a stratified clause set in which the atoms $x \approx t$ in the consequents of implications are subject to certain restrictions (e.g.\ $t \neq f(\ldots)$ and guardedness of atoms $u \approx c$ and $u \approx v$), then the above flattening rules yield clauses that belong to the following class of \BsrSli\ clauses---after necessary purification and normalization steps.
In the definition we mark certain predicate symbols that are intended to represent uninterpreted functions.
By adding suitable axioms later on, these will be equipped with the properties of function graphs.
\begin{definition}[Stratified and Guarded \BsrSli]\label{definition:StratifiedClauseSet}
	Consider a \BsrSli\ clause set $N$ in normal form.
	Let $R_1, \ldots, R_n$ be a list of predicate symbols that we consider to be \emph{marked in $N$}.
	We call $N$ \emph{stratified and guarded} with respect to $R_1, \ldots, R_n$, if and only if the following conditions are met.
	
	\begin{enumerate}[label=(\alph{*}), ref=\alph{*}]
		\item\label{enum:StratifiedClauseSet:I}
			There is some function $\lvl_N : \Pi \times \Nat \to \Nat$ that assigns to each argument position pair $\<P,i\>$ a nonnegative integer $\lvl_N\<P,i\>$ such that
			\begin{enumerate}[label=(\ref{enum:StratifiedClauseSet:I}.\arabic{*}), ref=(\ref{enum:StratifiedClauseSet:I}.\arabic{*})]
				\item\label{enum:StratifiedClauseSet:I:I} $\<P,i\> \prop_N \<Q,j\>$ entails $\lvl_N\<P,i\> = \lvl_N\<Q,j\>$, and
	
				\item\label{enum:StratifiedClauseSet:I:II} for every marked predicate symbol $R_j : \xi_1 \times \ldots \times \xi_m \times \xi_{m+1}$ we have $\lvl_N\<R_j,i\> > \lvl_N\<R_j,m+1\>$ for every $i \leq m$.
			\end{enumerate}
			
		\item\label{enum:StratifiedClauseSet:III}
			In every clause $\Lambda \,\|\, \Gamma \rightarrow \Delta$ in $N$ any occurrence of an atom $R_j(s_1, \ldots, s_m, v)$ in $\Delta$ entails that $\Gamma$ contains some atom $R_\ell(t_1, \ldots, t_{m'}, v)$.
			
		\item\label{enum:StratifiedClauseSet:II}
			For every atom $u \approx t$ in $N$, where $t$ is either a free-sort variable $v$ or a free-sort constant symbol, at least one of two cases applies:
			\begin{enumerate}[label=(\ref{enum:StratifiedClauseSet:II}.\arabic{*}), ref=(\ref{enum:StratifiedClauseSet:II}.\arabic{*})]
				\item\label{enum:StratifiedClauseSet:II:I} $u \approx t$, which must occur in the consequent of a clause, is guarded by some atom $R_j(t_1, \ldots, t_m, u)$ occurring in the antecedent of the same clause.
				\item\label{enum:StratifiedClauseSet:II:II} For every marked predicate symbol $R_j : \xi_1 \times \ldots \times \xi_m \times \xi_{m+1}$ and every argument position closure $\downcl_N\<R_j, i\>$ with $1 \leq i \leq m$ we have $\downcl_N\<R_j,i\> \cap \downcl_N(u) = \emptyset$. If $t = v$, we in addition have $\downcl_N\<R_j,i\> \cap \downcl_N(v) = \emptyset$.
			\end{enumerate}			
	\end{enumerate}
\end{definition}
Notice that any atom $u \approx v$ over distinct variables requires two guards $R(\bar{s}, u)$ and $R(\bar{t}, v)$ in order to be guarded in accordance with Condition~\ref{enum:StratifiedClauseSet:II:I}.


Let $N$ be a finite \BsrSli\ clause set in normal form that is stratified and guarded with respect to $R_1, \ldots, R_n$.
Let $R_i: \xi_1 \times \ldots \times \xi_m \times \xi_{m+1}$ be marked in $N$ and let $P: \zeta_1 \times \ldots \times \zeta_{m'}$ be any predicate symbol occurring in $N$ (be it marked or not).
We write $R_i \succeq P$ if and only if $\lvl_N\<R_i, m+1\> \geq \min_{1 \leq \ell \leq m'}\bigl( \lvl_N\<P, \ell\> \bigr)$.
Without loss of generality, we assume $R_1 \succeq_N \ldots \succeq_N R_n$.
Let
$\Phi_1 := \{ \forall x_1 \ldots x_m u u'. ( R_i(x_1, \ldots, x_m, u) \wedge R_i(x_1, \ldots, x_m, u') ) \rightarrow u \simeq u' \mid \text{$R_i$ has arity $m+1$} \}$ 
and
$\Phi_2 := \{ \forall x_1 \ldots x_m \exists u.\, R_i(x_1, \ldots, x_m, u) \mid \text{$R_i$ has arity}$ $m+1 \}$,
where~``$\simeq$'' is a placeholder for~``$\approx$'' in free-sort equations and for ``$=$'' in base-sort equations.

Given a set $M$ of \BsrSli\ clauses and an $(m+1)$-ary predicate symbol $R$ that is marked in $M$, we define the set
	$\hPhi(R, M) :=$\\
	$\bigl\{ R(c_1, \ldots, c_m, d_{R c_1 \ldots c_m}) \bigm| \<c_1, \ldots, c_m\> \in \cI_{\downcl_M\<R,\cdot\>}^{[m]} \bigr\}$ \\
	$\cup \bigl\{ \forall x_1 \ldots x_m. \bigvee_{\<c_1, \ldots, c_m\> \in \cI_{\downcl_M \<R,\cdot\>}^{[m]}}  R(x_1, \ldots, x_m, d_{R c_1 \ldots c_m}) \bigl\}$ \\
	$\cup \bigl\{ \forall x_1 \ldots x_m u.\; R(x_1, \ldots, x_m, u) \rightarrow \bigvee_{\<c_1, \ldots, c_m\> \in \cI_{\downcl_M \<R,\cdot\>}^{[m]}}  u \simeq d_{R c_1 \ldots c_m} \bigr\}$ \\
	$\cup \bigl\{ \forall x_1 \ldots x_m .\; R(x_1, \ldots, x_m, d_{R c_1 \ldots c_m}), R(x_1, \ldots, x_m, d_{R c'_1 \ldots c'_m})$ \\
	\strut\hspace{17ex} 
		$\rightarrow d_{R c_1 \ldots c_m} \simeq d_{R c'_1 \ldots c'_m} \bigm| \<c_1, \ldots, c_m\>, \<c'_1, \ldots, c'_m\> \in \cI_{\downcl_M\<R,\cdot\>}^{[m]} \bigr\}$\\
where $\cI_{\downcl_M\<R,\cdot\>}^{[m]}$ is used as an abbreviation for $\cI_{\downcl_M\<R,1\>} \times \ldots \times \cI_{\downcl_M\<R,m\>}$ and the $d_{R c_1 \ldots c_m}$ are assumed to be fresh constant symbols.
It is worth noticing that the clauses corresponding to $\hPhi(R,M)$ are stratified and guarded \BsrSli\ clauses.

We construct the sequence $M_0, M_1, \ldots, M_n$ of finite clause sets as follows:
$M_0 := N$,
every $M_{\ell+1}$ with $\ell \geq 0$ is an extension of $M_\ell$ by the \BsrSli\ clauses that correspond to the sentences in $\hPhi(R_{\ell+1}, M_\ell)$.

\begin{lemma}\label{lemma:StratifiedClauseSetsIntoBSR}
	The (finite) set $N \cup \Phi_1 \cup \Phi_2$ is satisfiable if and only if $M_n$ is satisfiable.
\end{lemma}
\begin{proof}[Proof sketch]
	Any hierarchic model of $\hPhi(R_1,M_0) \cup \ldots \cup \hPhi(R_n,M_{n-1})$ is also a hierarchic  model of $\Phi_1 \cup \Phi_2$.
	Hence, any hierarchic model of $M_n$ is also a hierarchic model of $N \cup \Phi_1 \cup \Phi_2$.
	Conversely, from any hierarchic model $\cA \models N \cup \Phi_1 \cup \Phi_2$ we can construct a hierarchic interpretation $\cB$ that is a model of both sets $N \cup \Phi_1 \cup \Phi_2$ and $M_n$, and for which the following set is finite for any $R: \xi_1 \times \ldots \times \xi_m \times \xi_{m+1}$ that is marked in $N$:\\
	\centerline{$\bigl\{ \fb \in \xi_{m+1}^\cB \mid \text{there are $\fa_1, \ldots, \fa_m$ such that $\<\fa_1, \ldots, \fa_m, \fb\> \in R^\cB$} \bigr\}$.}
	We develop the details of this construction in the proof of Lemma~\ref{lemma:MarkedPredicatesAsFunctionGraphsInFEU} in the appendix.
	Having $\cB$, we show that $\cB \models N \cup \Phi_1 \cup \Phi_2$ (Lemma~\ref{lemma:MarkedPredicatesAsFunctionGraphsInFEU}) and that $\cB \models M_n$ (Lemma~\ref{lemma:ArtificialAndExistentialInstantiationPoints}).
\end{proof}
This lemma entails that all the instantiation methods developed in Section~\ref{section:DecidingBSRwithSLIC} can be used to decide satisfiability of stratified and guarded \BsrSli\ clause sets.
 
 \begin{remark}
	In the definition of the sets $\hPhi(R,M)$ we refrained from optimizing the number of instantiation points by means of using $\cI_{\upcl_{M}\<P,i\>}$ instead of $\cI_{\downcl_{M}\<R,i\>}$ where this would lead to fewer instances.
	It is clear however, that this sort of optimization is compatible with the taken approach.
\end{remark}


We can add another background theory to the stratified and guarded fragment of \BsrSli\ while preserving compatibility with our instantiation approach.
Let $\Pi_\cT$ and $\Omega_\cT$ be finite sets of sorted predicate symbols and sorted function symbols, respectively, and let $\cT$ be some theory over $\Pi_\cT$ and $\Omega_\cT$.
We assume that $\Pi_\cT$ is disjoint from the set $\Pi$ of uninterpreted predicate symbols.
For any set $X$ of variables, let $\Tterms(X)$ be the set of all well-sorted terms constructed from the variables in $X$ and the function and constant symbols in $\Omega_\cT$.

\begin{definition}[\BsrSliT]\label{definition:RestrictedBSRwithConstrSyntax}	
	A clause set $N$ belongs to \emph{\BsrSliT} if it complies with the syntax of a \BsrSli\ clause set that is stratified and guarded with respect to certain predicate symbols $R_1, \ldots, R_n$ with the following exceptions.
	Let $C := \Lambda \,\|\, \Gamma \rightarrow \Delta$ be a clause in $N$.
	We allow atoms $P(s_1, \ldots, s_m)$ with $P \in \Pi_\cT$ and $s_1, \ldots, s_m \in \Tterms(V_\cZ \cup V_\cS)$---including equations $s_1 \approx s_2$---, if for every variable $u$ occurring in any of the $s_i$ there is either a LIA guard of the form $u = t$ in $\Lambda$ with $t$ being ground, or there is a guard $R_j(t_1, \ldots, t_{m'}, u)$ in $\Gamma$.
\end{definition}

The instantiation methods presented in Section~\ref{section:DecidingBSRwithSLIC} are also applicable to \BsrSliT, since Lemma~\ref{lemma:StratifiedClauseSetsIntoBSR} can be extended to cover finite \BsrSliT\ clause sets.
When computing instantiation points for \BsrSliT\ clause sets, we ignore $\cT$-atoms.
For example, a clause $\| R(t,u), P(s,c) \rightarrow P(s',u), Q(u)$ where $P(s,c)$ and $P(s',u)$ are $\cT$-atoms, \emph{does not} lead to an instantiation point $c$ for $\downcl\<Q,1\>$.
If we stick to this approach, the proof of Lemma~\ref{lemma:StratifiedClauseSetsIntoBSR} can easily be adapted to handle additional $\cT$-atoms. The involved model construction remains unchanged. $\cT$-atoms are basically treated like guarded free-sort atoms $u \approx d$.

\begin{proposition}
	\BsrSliT\ allows an (un)satisfiability-preserving embedding of the \emph{array property fragment} with in\-teger-indexed arrays and element theory $\cT$ (cf.~\cite{Bradley2006}) and of the \emph{finite essentially uninterpreted fragment} extended with simple integer arithmetic literals (cf.~\cite{Ge2009}) into \BsrSliT.
\end{proposition}

\begin{example}
	The following formula $\varphi$ belongs to the array property fragment with integer indices and the theory of bit vectors as the element theory.
	The operator $\sim$ stands for bitwise negation of bit vectors and the relations $\preceq$ and $\approx$ are used as the ``at most'' and the equality predicate on bit vectors, respectively.
	Moreover, $a[i]$ denotes a read operation on the array $a$ at index $i$. \\
	\centerline{$
		\begin{array}{r@{\hspace{2ex}}r@{\hspace{2ex}}c@{\hspace{2ex}}rr@{\;\;}l}
			\varphi := 	&c \geq 1 	&\wedge &\forall i j.	&0 \leq i \leq j	&\rightarrow\;\; a[i] \preceq a[j] \\
					&		&\wedge &\forall i.		&0 \leq i \leq c-1  	&\rightarrow\;\; a[i] \preceq {\sim\! a[0]} \\
					&		&\wedge &			&			&\rightarrow\;\; a[c] \approx {\sim\! a[0]} \\
					&		&\wedge &\forall i.		&i \geq c+1  		&\rightarrow\;\; a[i] \succeq {\sim\! a[0]} 
		\end{array}
	$}
	Translating $\varphi$ into \BsrSliT\ yields the following clause set $N$, in which we consider $P_a$ to be marked. \\
	\centerline{
		\begin{tabular}{c@{\hspace{4ex}}|@{\hspace{4ex}}c}
			$\begin{array}{r@{\;\;\|\;\;}r@{\;\;\rightarrow\;\;}l}
				c < 1 			& 				& \Box \\
				e \neq c-1		& 				& \Box \\
				f \neq c+1		& 				& \Box \\
			\end{array}$
			&			
			$\begin{array}{r@{\;\;\|\;\;}r@{\;\;\rightarrow\;\;}l}
				0 \leq i, i \leq j		& P_a(i, u), P_a(j,v)	& u \preceq v \\
				0 \leq i, i \leq e , y = 0	& P_a(i, u), P_a(y,v)	& u \preceq {\sim\! v} \\
				x = c, y = 0			&  P_a(x, u), P_a(y,v)	& u \approx {\sim\! v} \\
				i \geq f, y = 0		&P_a(i, u), P_a(y,v)	& u \succeq {\sim\! v}
			\end{array}$
		\end{tabular}	
	}
	In order to preserve (un)satisfiability, functional axioms have to be added for $P_a$ (cf.\ the sets $\Phi_1$ and $\Phi_2$ that we used earlier).
	Doing so, we leave \BsrSliT.
	
	The clause set $N$ induces the set $\cI_{\downcl\<P_a, 1\>} = \{c_{-\infty}, 0, c, f\}$ of instantiation points for the index of the array. 
	An adaptation of Lemma~\ref{lemma:StratifiedClauseSetsIntoBSR} for \BsrSliT\ entails that adding the clause set $N'$ corresponding to the following set of sentences yields a \BsrSliT\ clause set $N \cup N'$ that is equisatisfiable to $\varphi$. \\	
	\centerline{$	
		\begin{array}{l}
			\bigl\{ P_a(c', d_{P_a c'}) \bigm| c' \in \{c_{-\infty}, 0, c, f\} \bigr\} \\
			\cup\; \bigl\{ \forall i. \bigvee_{c' \in \{c_{-\infty}, 0, c, f\}}  P_a(i, d_{P_a c'}) \bigl\} \\
			\cup\; \bigl\{ \forall i u.\; P_a(i, u) \rightarrow \bigvee_{c' \in \{c_{-\infty}, 0, c, f\}}  u \approx d_{P_a c'} \bigr\} \\
			\cup\; \bigl\{ \forall i .\; P_a(i, d_{P_a c'}), P_a(i, d_{P_a c''}) \rightarrow d_{P_a c'} \approx d_{P_a c''} \bigm| c',c'' \in \{c_{-\infty}, 0, c, f\} \bigr\} 
		\end{array}
	$}	
	Using the instantiation methods that we have developed in Sections~\ref{section:BaseSortInstantiation} -- \ref{section:FreeSortInstantiation}, the set $N \cup N'$ can be turned into an equisatisfiable quantifier-free clause set.
	One possible (uniform) model $\cA \models N \cup N'$ assigns $c_{-\infty}^\cA = -1$, $e^\cA = 2$, $c^\cA = 3$, $f^\cA = 4$, $d_{P_a c_{-\infty}}^\cA = 00$, $d_{P_a 0}^\cA = 01$, $d_{P_a e}^\cA = 01$, $d_{P_a c}^\cA = 10$, $d_{P_a f}^\cA = 11$, and yields the array $\< 01, 01, 01, 10, 11, 11, 11, \ldots\>$.
	\qed
\end{example}
In the original array property fragment~\cite{Bradley2006} no nestings of array read operations are allowed. 
The stratification criterion in \BsrSliT\ prevents nestings of the form $a[a[i]]$, but it does not prevent nestings of the form $a[b[i]]$ with $a \neq b$.
In this sense, but not only in this sense, our fragment allows more freedom in formulating properties of arrays than the original array property fragment.


\section{Discussion}\label{section:conclusion}

We have demonstrated how universally quantified variables in \BsrSli\ clause sets can be instantiated economically.
In certain cases our methods lead to exponentially fewer instances than a naive instantiation with all occurring integer terms would generate.
Moreover, we have sketched how defining suitable finite-domain sort predicates instead of explicitly instantiating variables can avoid immediate blow-ups caused by explicit instantiation.
It is then left to the theorem prover to actually instantiate variables as needed.

We have shown that our methods are compatible with uninterpreted, non-constant functions under certain restrictions.
Even another background theory $\cT$ may be added, leading to \BsrSliT.
This entails applicability of our instantiation approach to known logic fragments, such as the \emph{array property fragment} \cite{Bradley2006}, the \emph{finite essentially uninterpreted fragment with arithmetic literals}~\cite{Ge2009}, and many-sorted first-order formulas over \emph{stratified vocabularies}~\cite{Abadi2010,Korovin2013}.

The instantiation methodology that we have described specifically for integer variables can also be adapted to work for universally quantified variables ranging over the reals \cite{VoigtWeidenbachArXiv2015}.
Our computation of instantiation points 
considers all argument positions in predicate atoms independently. 
This can be further refined by considering
dependencies between argument positions and clauses. 
For example, this refinement idea was successfully applied in first-order logic \cite{ClaessenLS11,Korovin2013}.

Once all the integer variables are grounded by successive instantiation, we are left with a clause set where for every integer variable $x$ in any clause there is a defining
equation $x = c$ for some constant $c$. 
Thus, the clause set can actually be turned into a standard first-order BSR clause
set by replacing the integer constants with respective fresh uninterpreted constants. 
Then, as an alternative to further grounding the free-sort variables, any state-of-the-art BSR decision procedure can be applied to test satisfiability~\cite{PiskacEtAl10,Korovin13,AlagiWeidenbach15}.
It is even sufficient to know the instantiation sets for the base sort variables. 
Then, instead of explicit grounding, by defining
respective finite-domain sort predicates for the sets, the worst-case exponential blow-up of grounding can be prevented, as outlined in Section~\ref{section:AvoidingNonlinearBlowups}.



\newpage

\appendix

\section{Appendix}

\subsection{Details Concerning Section~\ref{section:BaseSortInstantiation}}
\subsubsection*{Proof of Lemma~\ref{lemma:EquisatisfiabilityBaseInstantiation}}
\begin{lemma*}
	Let $N$ be a finite \BsrSli\ clause set in normal form such that, if the constant symbol $c_{-\infty}$ occurs in $N$, then $\Psi^{-\infty}_{N} \subseteq N$.
	Suppose there is a clause $C$ in $N$ which contains a base-sort variable $x$. Let $\hN_x$ be the clause set 
		$\hN_x := \bigl(N\setminus\{ C \}\bigr)  \;\cup\; \bigl\{C\subst{x/c}	\;\bigm|\; c\in\cI_{\downcl_N(x)} \bigr\} \;\cup\; \Psi^{-\infty}_{N}$.
	$N$ is satisfiable if and only if $\hN_x$ is satisfiable. 
\end{lemma*}
\begin{proof}
	The ``only if''-part is trivial.
			
	The ``if''-part requires a more sophisticated argument.
		In what follows, the notations $\prop$ and $\downcl$ always refer to the original clause set $N$.
		Let $\cA$ be a hierarchic model of $\hN_x$.
		We use $\cA$ to construct the hierarchic model $\cB$ as follows.
		For the domain $\cS^\cB$ we reuse $\cA$'s free domain $\cS^\cA$. 
		For all base-sort and free-sort constant symbols $c\in \consts(N)$, we set $c^{\cB} := c^{\cA}$.
		For every predicate symbol $P:\xi_1\times\ldots\times\xi_m \in \Pi$ that occurs in $N$, for every argument position $i$, $1\leq i\leq m$, with $\xi_i = \cZ$, and for every interval $p\in \cP_{\downcl\<P,i\>}^\cA$ Lemma~\ref{lemma:PartitionRepresentatives} and the extra clauses in $\Psi_N^{-\infty}$ guarantee the existence of a base-sort constant symbol $c_{\downcl\<P,i\>, p}\in \cI_{\downcl(x)}$, such that $c_{\downcl\<P,i\>, p}^\cA \in p$.
		Based on this observation, we define the family of projection functions $\pi_{\downcl\<P,i\>} : \Int\cup\cS^\cB \to \Int\cup\cS^\cA$ by
					\[ \pi_{\downcl\<P,i\>}(\fa) :=
						\begin{cases}
							c_{\downcl\<P,i\>, p}^\cA		&\text{if $\xi_i = \cZ$ and $p\in \cP_{\downcl\<P,i\>}^\cA$}\\
											&\text{is the interval $\fa$ lies in,} \\
							\fa				&\text{if $\xi_i = \cS$.}
						\end{cases}
					\]
		Using the projection functions $\pi_{\downcl\<P,i\>}$, we define the sets $P^\cB$ so that for all domain elements $\fa_1, \ldots, \fa_m$ of appropriate sorts $\bigl\<\fa_1, \ldots, \fa_m\bigr\> \in P^{\cB}$ if and only if $\bigl\<\pi_{\downcl\<P,1\>}(\fa_1), \ldots,$ $\pi_{\downcl\<P,m\>}(\fa_m)\bigr\> \in P^\cA$.

		We next show $\cB\models N$. Consider any clause $C' := \Lambda' \;\|\; \Gamma' \to \Delta'$ in $N$ and let $\beta : V_\cZ\cup V_\cS\to \Int\cup \cS^\cB$ be an arbitrary variable assignment. 
		From $\beta$ we derive a special variable assignment $\beta_\pi$ for which we shall infer $\cA,\beta_\pi \models C'$ as an intermediate step:
			$\beta_\pi(v) := \pi_{\downcl(v)}(\beta(v))$
		for every variable $v$.
		If $C' \neq C$, then $\hN_x$ already contains $C'$, and thus $\cA,\beta_\pi\models C'$ must hold.
		In case of $C' = C$, let $p_*$ be the interval in $\cP_{\downcl(x)}^\cA$ containing the value $\beta(x)$, and let $c_*$ be an abbreviation for $c_{\downcl(x), p_*}$. Due to $\beta_\pi(x) = c_*^\cA$ and since $\cA$ is a model of the clause $C\subst{x/c_*}$ in $\hN_x$, we conclude $\cA,\beta_\pi \models C$. 
		Hence, in any case we can deduce $\cA,\beta_\pi\models C'$. 
		By case distinction on why $\cA,\beta_\pi \models C'$ holds, we may use this result to infer $\cB,\beta\models C'$.
		\begin{description}
			\item Case $\cA, \beta_\pi \not\models s\triangleleft t$ for some ground atomic constraint $s\triangleleft t$ in $\Lambda'$. Since $\cB$ and $\cA$ interpret constant symbols in the same way and independently of a variable assignment, we immediately get $\cB, \beta\not\models s\triangleleft t$.
			
			\item Case $\cA,\beta_\pi \not\models (y \trianglelefteq d) \in \Lambda'$ for some base-sort variable $y$, some constant symbol $d$, and $\trianglelefteq {\in} \{\leq,$ $=, \geq\}$. This means $\beta_\pi(y) \mathrel{\not\!\trianglelefteq} d^\cA$. Let $p$ be the interval from $\cP_{\downcl(y)}^\cA$ that contains $\beta(y)$ and therefore also $\beta_\pi(y)$. 
				\begin{description}
					\item If $d^\cA$ lies outside of $p$, then $\beta_\pi(y) \trianglelefteq d^\cA$ if and only if $\beta(y) \trianglelefteq d^\cA$, since $\beta_\pi(y) \in p$ and $\beta(y) \in p$. Thus, $d^\cB = d^\cA$ entails $\cB, \beta \not\models y \trianglelefteq d$.
					
					\item If $p$ is the point interval $p = \{d^\cA\}$, then $\beta(y) = \beta_\pi(y) = d^\cA$, and thus $\cB,\beta \not\models y \trianglelefteq d$.
					
					\item Suppose $p = [r_\ell, r_u]$ and $r_\ell < d^\cA \leq r_u$, then $\trianglelefteq {\neq} \leq$, since $\beta_\pi(y) = c_{\downcl(y), p}^\cA = r_\ell < d^\cA$ (by Lemma~\ref{lemma:PartitionRepresentatives}). 
					Moreover, we conclude $d\not\in \cI_{\downcl(y)}$, since otherwise $p$ would be of the form $p = [d^\cA, r_u]$ by the construction of $\cP_{\downcl(y)}^\cA$. 
					Therefore, $\trianglelefteq {\not\in} \{\equals, \geq\}$, since otherwise the instantiation point $d$ would be in $\cI_{\downcl(y)}$. 
					But this contradicts our assumption that $\trianglelefteq {\in} \{\leq, \equals, \geq\}$.
					
						The case $p = [r_\ell, +\infty)$ with $r_\ell < d^\cA$ can be handled by similar arguments.

					\item Suppose $p = [d^\cA, r_u]$ and $d^\cA < r_u$, then $\beta_\pi(y) = c_{\downcl(y), p}^\cA = d^\cA$ by Lemma~\ref{lemma:PartitionRepresentatives}. 
						Consequently, \mbox{$\trianglelefteq {\not\in} \{\leq, \equals, \geq\}$}. 
						This contradicts the assumptions we made regarding the syntax of the constraint $y \trianglelefteq d$.
					
						The same applies in the case $p = [d^\cA, +\infty)$.
						
					\item Suppose $p = (-\infty, r_u]$ with $d^\cA \leq r_u$ or $p = (-\infty, +\infty)$.
						We know $c_{-\infty}^\cA \in p$ due to the extra clauses in $\hN_x$.
						\begin{description}
							\item If $c_{-\infty}^\cA = d^\cA$, then $d = c_{-\infty}$.
								Since we also have $\beta_\pi(y) = c_{-\infty}^\cA$, $\trianglelefteq$ cannot be one of the relations $\leq, \equals, \geq$.
							
							\item If $c_{-\infty}^\cA \neq d^\cA$, the fact that $d^\cA$ lies within $p$ entails that $d$ does not belong to $\cI_{\downcl(y)}$.
								Hence, $\triangleleft {\not\in} \{ \equals, \geq \}$.
								Therefore, we observe $\beta_\pi(y) > d^\cA$.
								But $\beta_\pi(y) = c_{-\infty}^\cA$ then leads to a contradiction with the clauses in $\Psi^{-\infty}_{N}$.
								
						\end{description}						
				\end{description}
				
			\item Case $\cA, \beta_\pi \not\models (y \leq z) \in \Lambda'$ for some base-sort variables $y, z$. This means $\beta_\pi(y) > \beta_\pi(z)$.
				Since $N$ is in normal form, we know that $\Gamma \to \Delta$ must contain atoms $P(\ldots, y, \ldots)$ and $R(\ldots, z, \ldots)$.
				By Lemma~\ref{lemma:BaseVariableInstantiationPointsClosure}, it follows that the partition $\cP_{\downcl(z)}^\cA$ is a refinement of $\cP_{\downcl(y)}^\cA$.
				
				Let $p_y = [r_\ell^y, r_u^y] \in \cP_{\downcl(y)}^\cA$ be the interval which contains $\beta(y)$ and let $p_z = [r_\ell^z, r_u^z] \in \cP_{\downcl(z)}^\cA$ be the interval which contains $\beta(z)$. 
				We distinguish several cases.
				\begin{description}
					\item If $\beta_\pi(z)$ lies outside of $p_y$, then $r_\ell^z = \beta_\pi(z) < \beta_\pi(y) = r_\ell^y$ together with the fact that $\cP_{\downcl(z)}^\cA$ is a refinement of $\cP_{\downcl(y)}^\cA$  implies $r_u^z < r_\ell^y$. Hence, $\beta(z) \in [r_\ell^z, r_u^z]$ and $\beta(y) \in [r_\ell^y, r_u^y]$ entail $\beta(z) < \beta(y)$ and thus $\cB, \beta \not\models y \leq z$. 
					
					\item Suppose $\beta_\pi(z)$ lies inside of $p_y$. Since $\cP_{\downcl(z)}^\cA$ is a refinement of $\cP_{\downcl(y)}^\cA$, we must have that $[r_\ell^z, r_u^z] \subseteq [r_\ell^y, r_u^y]$. But then $\beta_\pi(y) = r_\ell^y \leq r_\ell^z = \beta_\pi(z)$ contradicts the observation that $\beta_\pi(y) > \beta_\pi(z)$.
					
					\item Cases where $p_y = [r_\ell^y, +\infty)$ or $p_z = [r_\ell^z, +\infty)$ can be handled similarly.
					
					\item Suppose $p_y$ is of the form $(-\infty, r_u^y]$ or $(-\infty, +\infty)$.
						In this case we have $\beta_\pi(y) = c_{-\infty}^\cA$.
						This contradicts the observation $\beta_\pi(y) > \beta_\pi(z)$.
						
					\item Suppose $p_z$ is of the form $(-\infty, r_u^z]$ or $(-\infty, +\infty)$.
						In this case we have $\beta_\pi(z) = c_{-\infty}^\cA$.
						Since $\cP_{\downcl(z)}^\cA$ is a refinement of $\cP_{\downcl(y)}^\cA$, we either have $p_z \subseteq p_y$ or $p_y$ does not overlap with $p_z$.
						The former contradicts previous observations. 
						Therefore, the latter must apply and $p_y$ must be of the form $[r_\ell^y, r_u^y]$ or $[r_\ell^y, +\infty)$. 
						Moreover, $p_z$ has the form $(-\infty, r_u^z]$ with $r_u^z < r_\ell^y$.
						But then we conclude $\beta(z) \leq r_u^z < r_\ell^y \leq \beta(y)$.
						This observation entails $\cB, \beta \not\models y \leq z$.								
			  	\end{description}

			\item Case $\cA, \beta_\pi \not\models s\approx s'$ for some free atom $s\approx s' \in \Gamma'$. Hence, $s$ and $s'$ are either free-sort variables or constant symbols of the free sort, which means they do not contain subterms of the base sort. Since $\cB$ and $\cA$ behave identical on free-sort constant symbols and $\beta(u) = \beta_\pi(u)$ for any variable $u\in V_\cS$, it must hold $\cB, \beta \not\models s\approx s'$.

			\item Case $\cA, \beta_\pi \models s\approx s'$ for some $s\approx s' \in \Delta'$. Analogous to the above case, $\cB, \beta \models s\approx s'$ holds.
			
			\item Case $\cA, \beta_\pi \not\models P(s_1, \ldots, s_m)$ for some free atom $P(s_1, \ldots, s_m)\in\Gamma'$. This means \\ $\bigl\<\cA(\beta_\pi)(s_1), \ldots, \cA(\beta_\pi)(s_m)\bigr\> \not\in P^\cA$.
				\begin{description}
					\item Every $s_i$ of the free sort is either a constant symbol or a variable. Thus, we have $\cA(\beta_\pi)(s_i) = \cB(\beta)(s_i) = \pi_{\downcl\<P,i\>}(\cB(\beta)(s_i))$, since free-sort constant symbols are interpreted in the same way by $\cA$ and $\cB$, and because $\beta_\pi(u) = \beta(u)$ for every free-sort variable $u$.					
					\item Every $s_i$ that is of the base sort must be a variable. Hence, $\cA(\beta_\pi)(s_i) = c_{\downcl\<P,i\>,p}^\cA = \pi_{\downcl\<P,i\>}(\cB(\beta)(s_i))$, where $p$ is the interval in $\cP_{\downcl\<P,i\>}^\cA$ which contains $\beta(s_i)$ (and thus also $\beta_\pi(s_i)$) and where we have $\downcl(s_i) = \downcl\<P,i\>$.
				\end{description}
			Put together, this yields $\bigl\<\pi_{\downcl\<P,1\>}(\cB(\beta)(s_1)), \ldots, \pi_{\downcl\<P,m\>}(\cB(\beta)(s_m))\bigr\> \not\in P^\cA$. But then, by construction of $\cB$, we have $\bigl\<\cB(\beta)(s_1), \ldots, \cB(\beta)(s_m)\bigr\> \not\in P^\cB$, which entails $\cB, \beta \not\models P(s_1, \ldots,$ $s_m)$.					
													
			\item Case $\cA, \beta_\pi \models P(s_1, \ldots, s_m)$ for some free atom $P(s_1, \ldots, s_m)\in\Delta'$. Analogous to the above case we conclude $\cB, \beta \models P(s_1, \ldots, s_m)$.
		\end{description}
		Altogether, we have shown $\cB\models N$.
\end{proof}

\subsection{Details Concerning Section~\ref{section:FreeSortInstantiation}}
\subsubsection*{Proof of Lemma~\ref{lemma:EquisatisfiabilityFreeInstantiation}}
\begin{lemma*}
	Let $N$ be a finite \BsrSli\ clause set in normal form.
	Suppose there is a clause $C$ in $N$ which contains a free-sort variable $u$. Let $\hN_u := \bigl(N\setminus\{ C \}\bigr) \;\cup\; \bigl\{C\subst{u/c} \;\bigm|\; c\in\cI_{\downcl_N(u)} \bigr\}$.
	$N$ is satisfiable if and only if $\hN_u$ is satisfiable. 
\end{lemma*}
\begin{proof}
	The ``only if''-part is trivial.
			
	Consider the ``if''-part.
		In what follows, the notations $\prop$ and $\downcl$ always refer to the original clause set $N$.
		Let $\cA$ be a hierarchic model of $\hN_u$. 
		We use $\cA$ to construct the hierarchic model $\cB$ as follows.
		For the domain $\cS^\cB$ we take the set $\{ \fa \in \cS^\cA \mid \text{$\fa = c^\cA$ for some $c \in \fconsts(N)$} \}$. 
		For all base-sort and free-sort constant symbols $c\in \consts(N)$, we set $c^{\cB} := c^{\cA}$.
		For every argument position closure $\downcl\<P,i\>$ we fix some \emph{default instantiation point} $d_{\downcl\<P,i\>} \in \cI_{\downcl\<P,i\>}$. 
		To this end, we choose an arbitrary constant symbol from $\cI_{\downcl\<P,i\>}$.
		We define the family of projection functions $\pi_{\downcl\<P,i\>} : \Int\cup\cS^\cB \to \Int\cup\cS^\cA$ by
			\[ \pi_{\downcl\<P,i\>}(\fa) :=
				\begin{cases}
					\fa					&\text{if $\xi_i = \cS$ and $\fa = c^\cA$ for some $c \in \cI_{\downcl\<P,i\>}$,} \\
					d_{\downcl\<P,i\>}^\cA		&\text{if $\xi_i = \cS$ and $\fa \neq c^\cA$ for every $c \in \cI_{\downcl\<P,i\>}$,} \\
					\fa					&\text{if $\xi_i = \cZ$.}
				\end{cases}
			\]
		Using the projection functions $\pi_{\downcl\<P,i\>}$, we define the sets $P^\cB$ so that for all domain elements $\fa_1, \ldots, \fa_m$ of appropriate sorts $\bigl\<\fa_1, \ldots, \fa_m\bigr\> \in P^{\cB}$ if and only if $\bigl\<\pi_{\downcl\<P,1\>}(\fa_1), \ldots,$ $\pi_{\downcl\<P,m\>}(\fa_m)\bigr\> \in P^\cA$.

		We next show $\cB\models N$. Consider any clause $C' := \Lambda' \;\|\; \Gamma' \to \Delta'$ in $N$ and let $\beta : V_\cZ\cup V_\cS\to \Int\cup \cS^\cB$ be an arbitrary variable assignment. 
		From $\beta$ we derive a special variable assignment $\beta_\pi$ for which we shall infer $\cA,\beta_\pi \models C'$ as an intermediate step: for every variable $v$ we set $\beta_\pi(v) := \pi_{\downcl(v)}(\beta(v))$.
		If $C' \neq C$, then $\hN_u$ already contains $C'$, and thus $\cA,\beta_\pi\models C'$ must hold.
		In case of $C' = C$, we know that there is some constant symbol $c \in \cI_{\downcl(u)}$ such that $\beta_\pi(u) = c^\cA$. Since $C\subst{u/c}$ is a clause in $\hN_u$, $\cA$ is a model of $C\subst{u/c}$ and thus we conclude $\cA,\beta_\pi \models C$. 
		Hence, in any case we can deduce $\cA,\beta_\pi\models C'$. 
		By case distinction on why $\cA,\beta_\pi \models C'$ holds, we may use this result to infer $\cB,\beta\models C'$.
		\begin{description}
			\item Case $\cA, \beta_\pi \not\models s\triangleleft t$ for some atomic constraint $s\triangleleft t$ in $\Lambda'$. Since $\cB$ and $\cA$ interpret constant symbols in the same way and since $\beta$ and $\beta_\pi$ assign identical values to all base-sort variables, we immediately get $\cB, \beta\not\models s\triangleleft t$.
						
			\item Case $\cA, \beta_\pi \not\models s\approx t$ for some free atom $s\approx s' \in \Gamma'$. 
				Since $C'$ is in normal form, $s$ and $s'$ must be constant symbols.
				$\cB$ and $\cA$ interpret constant symbols in the same way and independently of a variable assignment and thus we immediately get $\cB, \beta\not\models s\approx t$.
			
			\item Case $\cA, \beta_\pi \models s\approx t$ for some $s\approx t \in \Delta'$. 
				\begin{description}
					\item If $s$ and $t$ are constant symbols, we know that $\cB, \beta \models s\approx t$ holds, by analogy to the above case.
					\item If $s$ is a free-sort variable $v$ and $t$ is a constant symbol $d$, we know that $d \in \cI_{\downcl(v)} = \fconsts(N)$ and thus $\beta_\pi(v) = d^\cA = \beta(v)$. 
						This entails $\cB,\beta \models v \approx d$.
						
					\item If $s$ is a free-sort variable $v$ and $t$ is a free-sort variable $w$, we know that $\cI_{\downcl(v)} = \cI_{\downcl(w)} = \fconsts(N)$ and thus $\beta(v) = \beta_\pi(v) = \beta_\pi(w) = \beta(w)$. 
						Consequently, we have $\cB,\beta \models v \approx w$.						
				\end{description}	

			\item Case $\cA, \beta_\pi \not\models P(s_1, \ldots, s_m)$ for some free atom $P(s_1, \ldots, s_m)\in\Gamma'$. This means $\bigl\<\cA(\beta_\pi)(s_1), \ldots,$ $\cA(\beta_\pi)(s_m)\bigr\> \not\in P^\cA$.
				\begin{description}
					\item Every $s_i$ that is of the base sort must be a variable. Hence, $\cA(\beta_\pi)(s_i) = \beta_\pi(s_i) = \beta(s_i) = \pi_{\downcl\<P,i\>}(\beta(s_i)) =  \pi_{\downcl\<P,i\>}(\cB(\beta)(s_i))$.
					
					\item Every $s_i$ of the free sort is either a constant symbol or a variable. 
					
						If $s_i$ is a constant symbol $d$, then we have $d \in \cI_{\downcl\<P,i\>}$.
						Hence, we have $\cA(\beta_\pi)(d) = d^\cA = \pi_{\downcl\<P,i\>}(d^\cA) = \pi_{\downcl\<P,i\>}(d^\cB) = \pi_{\downcl\<P,i\>}(\cB(\beta)(d))$ .
						
						If $s_i$ is a variable $v$, then\\ 
							\centerline{$\cA(\beta_\pi)(v) = \beta_\pi(v) = \pi_{\downcl(v)}(\beta(v)) = \pi_{\downcl\<P,i\>}(\cB(\beta)(v))$.}
				\end{description}
			Put together, this yields $\bigl\<\pi_{\downcl\<P,1\>}(\cB(\beta)(s_1)), \ldots, \pi_{\downcl\<P,m\>}(\cB(\beta)(s_m))\bigr\> \not\in P^\cA$. But then, by construction of $\cB$, we have $\bigl\<\cB(\beta)(s_1), \ldots, \cB(\beta)(s_m)\bigr\> \not\in P^\cB$, which entails $\cB, \beta \not\models P(s_1, \ldots,$ $s_m)$.					
													
			\item Case $\cA, \beta_\pi \models P(s_1, \ldots, s_m)$ for some free atom $P(s_1, \ldots, s_m)\in\Delta'$. Analogously to the above case we conclude $\cB, \beta \models P(s_1, \ldots, s_m)$.
		\end{description}
		Altogether, we have shown $\cB\models N$.
\end{proof}

\subsection{Details Concerning Section~\ref{section:StratifiedClauseSets}}

\begin{lemma}\label{lemma:MarkedPredicatesAsFunctionGraphsInFEU}
	Let $N$ be a clause set in normal form and let $N$ be stratified and guarded with respect to $R_1, \ldots, R_n$.
	Let $N'$ be the clause set that we obtain from $N$ by adding the clauses corresponding to the following sets of sentences:
	\begin{align*} 
		\Phi_1 := \bigl\{ \forall x_1 \ldots x_m u. \bigl( R(x_1, \ldots, x_m, &u) \wedge R(x_1, \ldots, x_m, u') \bigr) \rightarrow u \approx u' \\
			&\bigm| \text{$R : \xi_1\times \ldots \times \xi_m \times \xi_{m+1}$  is marked in $N$} \bigr\}
	\end{align*}
	and
	\begin{align*} 
		\Phi_2 := \bigl\{ \forall x_1 \ldots x_m &\exists u.\, R(x_1, \ldots, x_m, u) 
			\bigm| \text{$R : \xi_1\times \ldots \times \xi_m \times \xi_{m+1}$  is marked in $N$} \bigr\} ~. 
	\end{align*}	
	If $N'$ is satisfiable, then there is a model $\cB$ of $N'$ such that the following set is finite for any $R: \xi_1 \times \ldots \times \xi_m \times \xi_{m+1}$:\\
		\centerline{$\{ \fb \in \xi_{m+1}^\cB \mid \text{there are $\fa_1, \ldots, \fa_m$ such that $\<\fa_1, \ldots, \fa_m, \fb\> \in R^\cB$} \}$.}
\end{lemma}
\begin{proof}
	Without loss of generality, we assume $R_1 \succeq_N R_2 \succeq_N \ldots \succeq_N R_n$.

	Let $\cA$ be a model of $N'$.
	For every $R: \xi_1 \times \ldots \times \xi_m \times \xi_{m+1}$ among the $R_1, \ldots, R_n$
	let $\tau_R^\cA: \xi_1^\cA \times \ldots \times \xi_m^\cA \to \xi_{m+1}^\cA$ be a mapping such that for every tuple $\<\fa_1, \ldots, \fa_m\>$ of domain elements we have \\
		\centerline{$\bigl\< \fa_1, \ldots, \fa_m, \tau_R^\cA(\fa_1, \ldots, \fa_m) \bigr\> \in R^\cA$.}
	Due to $\cA \models \Phi_1 \cup \Phi_2$, every $\tau_R^\cA$ is uniquely determined.
	
	In the rest of the proof $\downcl$ is an abbreviation for $\downcl_N$ and $\prop$ stands for $\prop_N$.
	
	Let $P$ be any predicate symbol occurring in $N$.
	We introduce artificial instantiation points as follows.
	Let $\hI_{\downcl\<P,i\>}$ be the smallest set satisfying the following requirements.
	\begin{enumerate}[label=($\hI$-\alph{*}), ref=($\hI$-\alph{*})]
		\item\label{enum:ArtificialInstantiationPoints:I} $\cI_{\downcl\<P,i\>} \subseteq \hI_{\downcl\<P,i\>}$.
		\item\label{enum:ArtificialInstantiationPoints:II} For every $R: \xi_1 \times \ldots \times \xi_m \times \xi_{m+1}$ that is marked in $N$ and for which $R \succeq P$ and $\<R,m+1\> \prop \<P,i\>$
			we have $d_{R c_1 \ldots c_m} \in \hI_{\downcl\<P,i\>}$ for all tuples $\<c_1, \ldots, c_m\> \in \hI_{\downcl\<R,1\>} \times \ldots \times \hI_{\downcl\<R,m\>}$.
		\item\label{enum:ArtificialInstantiationPoints:III} If there is some free-sort atom $u \approx t$ ($t$ being ground or non-ground) in $N$ that is \emph{not guarded} (cf.\ Condition~(\ref{enum:StratifiedClauseSet:II}) in Definition~\ref{definition:StratifiedClauseSet}) and for which $\downcl(u) = \downcl\<P,i\>$, then $\hI_{\downcl\<Q,j\>} \subseteq \hI_{\downcl\<P,i\>}$ for every argument position pair $\<Q,j\>$.
		
			In other words, in this case $\hI_{\downcl\<P,i\>}$ collects all artificial instantiation points that are introduced into any set $\hI_{\downcl\<Q,j\>}$.
	\end{enumerate}
	The $d_{R c_1 \ldots c_m}$ are assumed to be fresh constant symbols that do not occur in $N$.
	Their intended meaning is fixed by assuming $d_{R c_1 \ldots c_m}^\cA := \tau_R^\cA (c_1^\cA, \ldots, c_m^\cA)$ (without loss of generality).
	Moreover, we assume that $c_{-\infty}$ does not occur in $N$ (but may occur as instantiation point) and we set the value of $c_{-\infty}$ so that $c_{-\infty}^\cA < c^\cA$ holds for every base-sort constant symbol $c$ occurring in $N$ and any $c$ that is an artificial instantiation point of the base sort.
			\begin{description}
				\item \underline{Claim:} 
						For every argument position closure $\downcl\<P,i\>$ the set $\hI_{\downcl\<P,i\>}$ is finite.
				\item \underline{Proof:} 
						All the $\cI_{\downcl\<Q,j\>}$ are finite, since $N$ and the clauses therein are assumed to be finite.
						Hence, if $\hI_{\downcl\<P,i\>}$ were infinite, then it would contain infinitely many artificial instantiation points.
						
						Consider any artificial instantiation point $d_{R c_1 \ldots c_{k-1} d_{R' c'_1 \ldots c'_{m'}} c_{k+1} \ldots c_m}$ with $R : \xi_1 \times \ldots \times \xi_m \times \xi_{m+1}$ and $R' : \zeta_1 \times \ldots \times \zeta_{m'} \times \zeta_{m'+1}$, both being marked in $N$.
						Hence, $d_{R' c'_1 \ldots c'_{m'}} \in \hI_{\downcl\<R,k\>} \setminus \cI_{\downcl\<R,k\>}$.
						
						Assume that $d_{R' c'_1 \ldots c'_{m'}}$ has been added to $\hI_{\downcl\<R,k\>}$ because of requirement~\ref{enum:ArtificialInstantiationPoints:III}.
						Hence, there is some free-sort variable $u$ such that $\downcl\<R,k\> = \downcl(u)$ and there is some unguarded free-sort atom $u \approx t$ in some clause in $N$.
						By Condition~\ref{enum:StratifiedClauseSet:II:II} of Definition~\ref{definition:StratifiedClauseSet}, $R$ cannot be marked in $N$.
						This contradicts our assumptions.
						
						Assume that $d_{R' c'_1 \ldots c'_{m'}}$ has been added to $\hI_{\downcl\<R,k\>}$ because of requirement~\ref{enum:ArtificialInstantiationPoints:II}.
						Consequently, we have $R' \succeq R$ and $\<R',m'+1\> \prop \<R,k\>$.
						The latter fact entails $\lvl_N\<R',m'+1\> = \lvl_N\<R,k\>$.\
						Since $N$ is stratified and $R$ marked in $N$, we must have $\lvl_N\<R,k\> > \lvl_N\<R,$ $m+1\>$.
						Hence, $\lvl_N\<R',m'+1\> > \lvl_N\<R,m+1\>$.
						
						This means, the length of chains of the form
							$d_1 = d_{R_{j_1} \ldots d_2 \dots}$,
							$d_2 = d_{R_{j_2} \ldots d_3 \dots}$,
							$\ldots$,
							$d_{k} = d_{R_{j_k} \ldots d_{k+1} \dots}$,
							$\ldots$
						is upper bounded by the highest level that $\lvl_N$ assigns to any argument position pair in $N$.
							
						Consequently, $\hI_{\downcl\<P,i\>}$ must be finite.						
						\strut\hfill$\Diamond$
			\end{description}
	
	We next define a family of projections $\pi_{\downcl\<P,i\>}$ for every predicate symbol $P : \zeta_1 \times \ldots \times \zeta_m$ occurring in $N$:
		\[
			\pi_{\downcl\<P,i\>}(\fa) := 
							\begin{cases}
								c_{\downcl\<P,i\>,p}^\cA		&\text{if $\zeta_i = \cZ$ and $p \in \hP_{\downcl\<P,i\>}^\cA$ is the interval $\fa$ lies in,}\\
								\fa					&\text{if $\zeta_i = \cS$ and $\fa = c^\cA$ for some $c \in \hI_{\downcl\<P,i\>}$,}\\
								d_{\downcl\<P,i\>}^\cA		&\text{if $\zeta_i = \cS$ and $\fa \neq c^\cA$ for every $c \in \hI_{\downcl\<P,i\>}$,}
							\end{cases}		
		\]
	where $\hP_{\downcl\<P,i\>}^\cA$ is defined based on $\hI_{\downcl\<P,i\>}$ (cf.\ Definition~\ref{definition:BaseVariableInstantiationPoints}), 
	$c_{\downcl\<P,i\>,p}$ is some constant symbol in $\hI_{\downcl\<P,i\>}$ such that $c_{\downcl\<P,i\>,p}^\cA \in p$,
	and $d_{\downcl\<P,i\>}$ is some \emph{default instantiation point} of sort $\cS$ picked from $\cI_{\downcl\<P,i\>}$ (not $\hI_{\downcl\<P,i\>}$).
	
	We are now ready to construct the hierarchic interpretation $\cB$:
	\begin{itemize}
		\item $\cS^\cB := \bigl\{ c^\cA \bigm| c \in \fconsts(N) \bigr\}$\\
		 \strut\hspace{7ex}$\cup \bigl\{ d_{R c_1 \ldots c_m}^\cA \bigm| \text{$d_{R c_1 \ldots c_m}$ is some free-sort artificial instantiation point} \bigr\}$,
		
		\item $c^\cB := c^\cA$ for every constant symbol occurring in $N$ and also for every artificially introduced instantiation point $d_{R c_1 \ldots c_m}$, i.e.\ $d_{R c_1 \ldots c_m}^\cB := \tau_R^\cA(c_1^\cA, \ldots, c_m^\cA)$,
		
		\item for every non-marked $Q : \zeta_1 \times \ldots \times \zeta_m$ occurring in $N$ and every tuple $\<\fa_1, \ldots, \fa_m\>$ of appropriate sort we set 
			$\<\fa_1, \ldots, \fa_m\> \in Q^\cB$ if and only if $\bigl\< \pi_{\downcl\<Q,1\>}(\fa), \ldots, \pi_{\downcl\<Q,m\>}(\fa_m) \bigr\> \in Q^\cA$,
			
		\item for every marked $R : \xi_1 \times \ldots \times \xi_m \times \xi_{m+1}$ occurring in $N$, every tuple $\<\fa_1, \ldots, \fa_m\>$ of appropriate sort, and any domain element $\fb$ we set 
			$\<\fa_1, \ldots, \fa_m, \fb\> \in R^\cB$ if and only if $\bigl\< \pi_{\downcl\<R,1\>}(\fa), \ldots, \pi_{\downcl\<R,m\>}(\fa_m), \fb \bigr\> \in R^\cA$.		
	\end{itemize}
	Notice that $\<\fa_1, \ldots, \fa_m, \fb\> \in R^\cB$ if and only if $\fb = \tau_R^\cA(\pi_{\downcl\<R,1\>}(\fa_1), \ldots, \pi_{\downcl\<R,m\>}(\fa_m))$ for every marked $R$, because of $\cA \models \Phi_1$. 
	Hence, the set\\
		\centerline{$\{ \fb \mid \text{there are $\fa_1, \ldots, \fa_m$ such that $\<\fa_1, \ldots, \fa_m, \fb\> \in R^\cB$} \}$}
	is finite.
	
	Next, we show $\cB \models N'$.
	The first observation that we make is that, due to $\cA \models \Phi_1 \cup \Phi_2$ and due to the construction of $\cB$, $\cB$ also satisfies $\Phi_1 \cup \Phi_2$.
	It remains to show that $\cB$ is a hierarchic model of $N$.
	
	Consider any clause $C := \Lambda \;\|\; \Gamma \to \Delta$ in $N$ and let $\beta : V_\cZ\cup V_\cS\to \Int\cup \cS^\cB$ be some variable assignment. 
	From $\beta$ we derive a special variable assignment $\beta_\pi$: for every variable $v$ we set $\beta_\pi(v) := \pi_{\downcl(v)}(\beta(v))$.
	By assumption, $\cA$ is a model of $C$ and thus we conclude $\cA,\beta_\pi \models C$.
	By case distinction on why $\cA,\beta_\pi \models C$ holds, we may use this result to infer $\cB,\beta\models C$.
	\begin{description}
		\item Case $\cA, \beta_\pi \not\models s\triangleleft t$ for some ground LIA constraint $s\triangleleft t$ in $\Lambda$. Since $\cB$ and $\cA$ interpret constant symbols in the same way and independently of a variable assignment, we immediately get $\cB, \beta\not\models s\triangleleft t$.
			
		\item Case $\cA,\beta_\pi \not\models (y \trianglelefteq d) \in \Lambda$ for some base-sort variable $y$, some constant symbol $d$, and $\trianglelefteq {\in} \{\leq,$ $=, \geq\}$. This means $\beta_\pi(y) \mathrel{\not\!\trianglelefteq} d^\cA$. Let $p$ be the interval from $\hP_{\downcl(y)}^\cA$ that contains $\beta(y)$ and therefore also $\beta_\pi(y)$. 
			\begin{description}
				\item If $d^\cA$ lies outside of $p$, then $\beta_\pi(y) \trianglelefteq d^\cA$ if and only if $\beta(y) \trianglelefteq d^\cA$, since $\beta_\pi(y) \in p$ and $\beta(y) \in p$. Thus, $d^\cB = d^\cA$ entails $\cB, \beta \not\models y \trianglelefteq d$.
				
				\item If $p$ is the point interval $p = \{d^\cA\}$, then $\beta(y) = \beta_\pi(y) = d^\cA$, and thus $\cB,\beta \not\models y \trianglelefteq d$.
				
				\item Suppose $p = [r_\ell, r_u]$ and $r_\ell < d^\cA \leq r_u$, then $\trianglelefteq {\neq} \leq$, since $\beta_\pi(y) = c_{\downcl(y), p}^\cA = r_\ell < d^\cA$ (by Lemma~\ref{lemma:PartitionRepresentatives}). 
				Moreover, we conclude $d\not\in \hI_{\downcl(y)}$, since otherwise $p$ would be of the form $p = [d^\cA, r_u]$ by the construction of $\hP_{\downcl(y)}^\cA$. 
				Therefore, $\trianglelefteq {\not\in} \{\equals, \geq\}$, since otherwise the instantiation point $d$ would be in $\hI_{\downcl(y)}$. 
				But this contradicts our assumption that $\trianglelefteq {\in} \{\leq, \equals, \geq\}$.
				
					The case $p = [r_\ell, +\infty)$ with $r_\ell < d^\cA$ can be handled by similar arguments.

				\item Suppose $p = [d^\cA, r_u]$ and $d^\cA < r_u$, then $\beta_\pi(y) = c_{\downcl(y), p}^\cA = d^\cA$ by Lemma~\ref{lemma:PartitionRepresentatives}. 
					Consequently, \mbox{$\trianglelefteq {\not\in} \{\leq, \equals, \geq\}$}. 
					This contradicts the assumptions we made regarding the syntax of the constraint $y \trianglelefteq d$.
				
					The same applies in the case $p = [d^\cA, +\infty)$.
					
				\item Suppose $p = (-\infty, r_u]$ with $d^\cA \leq r_u$ or $p = (-\infty, +\infty)$.
					We know $c_{-\infty}^\cA \in p$ due to our earlier assumption on the value that is assigned to $c_{-\infty}$ by $\cA$.
					By the same assumption, we know that $c_{-\infty}^\cA < d^\cA$.
					The fact that $d^\cA$ lies within $p$ entails that $d$ does not belong to $\hI_{\downcl(y)}$.
					Hence, $\triangleleft {\not\in} \{ \equals, \geq \}$.
					Therefore, we conclude $\beta_\pi(y) > d^\cA$.
					But $\beta_\pi(y) = c_{-\infty}^\cA < d^\cA$ then leads to a contradiction.						
			\end{description}
			
		\item Case $\cA, \beta_\pi \not\models (y \leq z) \in \Lambda$ for two base-sort variables $y, z$. This means $\beta_\pi(y) > \beta_\pi(z)$.
			\begin{description}
				\item \underline{Claim:} 
						$\hI_{\downcl(y)} \subseteq \hI_{\downcl(z)}$.
				\item \underline{Proof:}
						Since $N$ is in normal form, we know that $\Gamma \to \Delta$ must contain atoms $P(\ldots, y, \ldots)$ and $Q(\ldots, z, \ldots)$ where $y$ occurs in the $i$-th argument position and $z$ in the $j$-th. Hence, we have $\<P,i\> \prop \<Q,j\>$ and thus also $\cI_{\downcl\<P,i\>} \subseteq \cI_{\downcl\<Q,j\>}$, by Lemma~\ref{lemma:BaseVariableInstantiationPointsClosure}.
						
						Suppose that $R: \xi_1 \times \ldots \times \xi_m \times \xi_{m+1}$ is some marked predicate symbol such that $R \succeq P$ and $\<R,m+1\> \prop \<P,i\>$.
						Since we assume $N$ to be stratified with respect to $R_1, \ldots, R_n$, $\<R,m+1\> \prop \<P,i\> \prop \<Q,j\>$ entails $\lvl_N\<R,m+1\> = \lvl_N\<P,i\> = \lvl_N\<Q,j\>$. 
						Consequently, we observe $R \succeq Q$ and $\<R,m+1\> \prop \<Q,j\>$, by transitivity of $\prop$.
						This means any artificial instantiation points that are introduced into $\hI_{\downcl\<P,i\>}$ because of $R$ are also introduced into $\hI_{\downcl\<Q,j\>}$.
						
						Therefore, we observe $\hI_{\downcl\<P,i\>} \subseteq \hI_{\downcl\<Q,j\>}$.																	\strut\hfill$\Diamond$
			\end{description}
			By virtue of the above claim, we conclude that $\hP_{\downcl(z)}^\cA$ is a refinement of $\hP_{\downcl(y)}^\cA$.
			
			Let $p_y = [r_\ell^y, r_u^y] \in \hP_{\downcl(y)}^\cA$ be the interval which contains $\beta(y)$ and let $p_z = [r_\ell^z, r_u^z] \in \hP_{\downcl(z)}^\cA$ be the interval which contains $\beta(z)$. 
			We distinguish several cases.
			\begin{description}
				\item If $\beta_\pi(z)$ lies outside of $p_y$, then $r_\ell^z = \beta_\pi(z) < \beta_\pi(y) = r_\ell^y$ together with the fact that $\hP_{\downcl(z)}^\cA$ is a refinement of $\hP_{\downcl(y)}^\cA$  implies $r_u^z < r_\ell^y$. Hence, $\beta(z) \in [r_\ell^z, r_u^z]$ and $\beta(y) \in [r_\ell^y, r_u^y]$ entail $\beta(z) < \beta(y)$ and thus $\cB, \beta \not\models y \leq z$. 
				
				\item Suppose $\beta_\pi(z)$ lies inside of $p_y$. Since $\hP_{\downcl(z)}^\cA$ is a refinement of $\hP_{\downcl(y)}^\cA$, we must have that $[r_\ell^z, r_u^z] \subseteq [r_\ell^y, r_u^y]$. But then $\beta_\pi(y) = r_\ell^y \leq r_\ell^z = \beta_\pi(z)$ contradicts the observation that $\beta_\pi(y) > \beta_\pi(z)$.
				
				\item Cases where $p_y = [r_\ell^y, +\infty)$ or $p_z = [r_\ell^z, +\infty)$ can be handled similarly.
				
				\item Suppose $p_y$ is of the form $(-\infty, r_u^y]$ or $(-\infty, +\infty)$.
					In this case we have $\beta_\pi(y) = c_{-\infty}^\cA$.
					This contradicts the observation $\beta_\pi(y) > \beta_\pi(z)$.
					
				\item Suppose $p_z$ is of the form $(-\infty, r_u^z]$ or $(-\infty, +\infty)$.
					In this case we have $\beta_\pi(z) = c_{-\infty}^\cA$.
					Since $\hP_{\downcl(z)}^\cA$ is a refinement of $\hP_{\downcl(y)}^\cA$, we either have $p_z \subseteq p_y$ or $p_y$ does not overlap with $p_z$.
					The former contradicts previous observations. 
					Therefore, the latter must apply and $p_y$ must be of the form $[r_\ell^y, r_u^y]$ or $[r_\ell^y, +\infty)$. 
					Moreover, $p_z$ has the form $(-\infty, r_u^z]$ with $r_u^z < r_\ell^y$.
					But then we conclude $\beta(z) \leq r_u^z < r_\ell^y \leq \beta(y)$.
					This observation entails $\cB, \beta \not\models y \leq z$.								
		  	\end{description}					

		\item Case $\cA, \beta_\pi \not\models Q(s_1, \ldots, s_m)$ for some free atom $Q(s_1, \ldots, s_m)\in\Gamma$ with $Q$ being unmarked. 
			This means  $\bigl\<\cA(\beta_\pi)(s_1), \ldots, \cA(\beta_\pi)(s_m)\bigr\> \not\in Q^\cA$.
			\begin{description}
				\item Every $s_i$ that is of the base sort must be a variable. 
					Hence, $\cA(\beta_\pi)(s_i) = \beta_\pi(s_i) = \pi_{\downcl\<P,i\>}(\beta(s_i)) =  \pi_{\downcl\<P,i\>}(\cB(\beta)(s_i))$.
				
				\item Every $s_i$ of the free sort is either a constant symbol or a variable. 
				
					If $s_i$ is a constant symbol $c$, then we have $c \in \cI_{\downcl\<P,i\>} \subseteq \hI_{\downcl\<P,i\>}$.
					Hence, we have $\cA(\beta_\pi)(c) = c^\cA = \pi_{\downcl\<P,i\>}(c^\cA) = \pi_{\downcl\<P,i\>}(c^\cB) = \pi_{\downcl\<P,i\>}(\cB(\beta)(c))$ .
					
					If $s_i$ is a variable $v$, then\\
						\centerline{$\cA(\beta_\pi)(v) = \beta_\pi(v) = \pi_{\downcl(v)}(\beta(v)) = \pi_{\downcl\<P,i\>}(\cB(\beta)(v))$.}
			\end{description}
		Put together, this yields $\bigl\<\pi_{\downcl\<P,1\>}(\cB(\beta)(s_1)), \ldots, \pi_{\downcl\<P,m\>}(\cB(\beta)(s_m))\bigr\> \not\in P^\cA$. But then, by construction of $\cB$, we have $\bigl\<\cB(\beta)(s_1), \ldots, \cB(\beta)(s_m)\bigr\> \not\in P^\cB$, which entails $\cB, \beta \not\models P(s_1, \ldots,$ $s_m)$.					
												
		\item Case $\cA, \beta_\pi \models Q(s_1, \ldots, s_m)$ for some free atom $q(s_1, \ldots, s_m)\in\Delta$ with unmarked $Q$. 
			Analogously to the above case we conclude $\cB, \beta \models P(s_1, \ldots, s_m)$.
			
		\item Case $\cA, \beta_\pi \not\models R(s_1, \ldots, s_m, t)$ for some free atom $R(s_1, \ldots, s_m, t)\in\Gamma$ with $R$ being marked in $N$. 
			This means\\
				\centerline{$\bigl\<\cA(\beta_\pi)(s_1), \ldots, \cA(\beta_\pi)(s_m), \cA(\beta_\pi)(t)\bigr\> \not\in R^\cA$.}
			Moreover, it follows $\cA(\beta_\pi)(t) \neq \tau_R^\cA\bigl( \cA(\beta_\pi)(s_1), \ldots, \cA(\beta_\pi)(s_m) \bigr)$.
			
			As in the previous case, we can show 
			\begin{itemize}
				\item[($*$)] $\cA(\beta_\pi)(s_i) = \pi_{\downcl\<R,i\>}(\cB(\beta)(s_i))$ for every $s_i$, $1 \leq i \leq m$.			
			\end{itemize}	
			\begin{description}
				\item If $t$ is a constant symbol $d$, then 
					$\cA(\beta_\pi)(d) = d^\cA = d^\cB = \cB(\beta)(d)$.
					Due to \\
						$d^\cA \neq \tau_R^\cA\bigl( \cA(\beta_\pi)(s_1), \ldots, \cA(\beta_\pi)(s_m) \bigr)$\\
						\strut\hspace{7ex} $= \tau_R^\cA \bigl( \pi_{\downcl\<R,1\>}(\cB(\beta)(s_1)), \ldots, \pi_{\downcl\<R,m\>}(\cB(\beta)(s_m)) \bigr)$,\\
					we have $\bigl\< \cB(\beta)(s_1), \ldots, \cB(\beta)(s_m), \cB(\beta)(d) \bigr\> \not\in R^\cB$.
															
				\item If $t$ is a variable $v$, then $\cA(\beta_\pi)(v) = \beta_\pi(v) = \pi_{\downcl(v)}(\beta(v))$.
					By definition of $\pi_{\downcl(v)}$, there must be some instantiation point $d \in \hI_{\downcl(v)}$ such that $\beta_\pi(v) = d^\cA$.					
					Similarly, by definition of the $\pi_{\downcl\<R,i\>}$, ($*$) entails the existence of a tuple of instantiation points $\<c_1, \ldots, c_m\> \in \hI_{\downcl\<R,1\>} \times \ldots \times \hI_{\downcl\<R,m\>}$ such that for every $i$, $1 \leq i \leq m$, we have $c_i^\cA = \cA(\beta_\pi)(s_i) = \pi_{\downcl\<R,i\>}(\cB(\beta)(s_i))$.
					Hence, by reflexivity of the relations $\succeq_N$ and $\prop$, we know that there is some artificial instantiation point $d_{R c_1 \ldots c_m} \in \hI_{\downcl\<R,m+1\>}$ such that $d_{R c_1 \ldots c_m}^\cA = \tau_R^\cA(c_1^\cA, \ldots, c_m^\cA)$.
					
					Because of $d^\cA = \cA(\beta_\pi)(d) \neq \tau_R^\cA(c_1^\cA, \ldots, c_m^\cA) = d_{R c_1 \ldots c_m}^\cA$, it follows that $d \neq d_{R c_1 \ldots c_m}$.
					Since $\pi_{\downcl\<R,m+1\>}$ projects $\beta(v)$ onto some value different from $d_{R c_1 \ldots c_m}^\cA$, the original $\beta(v)$ must be different from $d_{R c_1 \ldots c_m}^\cA$.
					Hence, \\
						\centerline{$\bigl\< \pi_{\downcl\<R,1\>}\bigl(\cB(\beta)(s_1)\bigr), \ldots, \pi_{\downcl\<R,m\>}\bigl(\cB(\beta)(s_m)\bigr), \cB(\beta)(t) \bigr\> = \bigl\< c_1^\cA, \ldots, c_m^\cA, \beta(v) \bigr\> \not\in R^\cA$ \hspace{11ex}}
					and thus also 
					$\bigl\< \cB(\beta)(s_1), \ldots \cB(\beta)(s_m), \cB(\beta)(t) \bigr\> \not\in R^\cB$.
			\end{description}
			Hence, we have $\cB, \beta \not\models R(s_1, \ldots, s_m, t)$.			

		\item Case $\cA, \beta_\pi \models R(s_1, \ldots, s_m, t)$ for some free atom $R(s_1, \ldots, s_m, t) \in \Delta$ with $R$ being marked in $N$. 
			This means\\
				\centerline{$\bigl\<\cA(\beta_\pi)(s_1), \ldots, \cA(\beta_\pi)(s_m), \cA(\beta_\pi)(t)\bigr\> \in R^\cA$.}
			Moreover, it follows $\cA(\beta_\pi)(t) = \tau_R^\cA\bigl( \cA(\beta_\pi)(s_1), \ldots, \cA(\beta_\pi)(s_m) \bigr)$.
			
			As in the previous case, we can show 
			\begin{itemize}
				\item[($*$)] $\cA(\beta_\pi)(s_i) = \pi_{\downcl\<R,i\>}(\cB(\beta)(s_i))$ for every $s_i$, $1 \leq i \leq m$.			
			\end{itemize}	
			\begin{description}
				\item If $t$ is a constant symbol $d$, then 
					$\cA(\beta_\pi)(d) = d^\cA = d^\cB = \cB(\beta)(d)$.
					Due to \\
						$d^\cA = \tau_R^\cA\bigl( \cA(\beta_\pi)(s_1), \ldots, \cA(\beta_\pi)(s_m) \bigr)$\\
							 	\strut\hspace{7ex} $= \tau_R^\cA \bigl( \pi_{\downcl\<R,1\>}(\cB(\beta)(s_1)), \ldots, \pi_{\downcl\<R,m\>}(\cB(\beta)(s_m)) \bigr)$,\\
					we have $\bigl\< \cB(\beta)(s_1), \ldots, \cB(\beta)(s_m), \cB(\beta)(d) \bigr\> \in R^\cB$.
															
				\item If $t$ is a variable $v$, then $\cA(\beta_\pi)(v) = \beta_\pi(v) = \pi_{\downcl(v)}(\beta(v))$.
					Since we assume $N$ to be guarded with respect to $R$, $\Gamma$ must contain an atom of the form $R'(t_1, \ldots, t_{m'}, v)$ with $R'$ being marked in $N$. 
					The case $\cA,\beta_\pi \not\models R'(t_1, \ldots, t_{m'}, v)$ has been treated earlier, and thus we assume $\cA,\beta_\pi \models R'(t_1, \ldots, t_{m'}, v)$.
					
					Similarly to ($*$), we can prove
					\begin{itemize}
						\item[($**$)] $\cA(\beta_\pi)(t_i) = \pi_{\downcl\<R',i\>}(\cB(\beta)(t_i))$ for every $t_i$, $1 \leq i \leq m'$.	
					\end{itemize}	
					
					By ($*$) and ($**$) we have \\
						\centerline{$\beta_\pi(v) 
									= \tau_{R}^\cA \bigl( \pi_{\downcl\<R,1\>}(\cB(\beta)(s_1), \ldots, \pi_{\downcl\<R,m\>}(\cB(\beta)(s_m)) \bigr)$}
						\centerline{\hspace{9ex}$		= \tau_{R'}^\cA \bigl( \pi_{\downcl\<R',1\>}(\cB(\beta)(t_1), \ldots, \pi_{\downcl\<R',m'\>}(\cB(\beta)(t_{m'})) \bigr)$.}
				
					We distinguish two cases.
					\begin{description}
						\item If $\beta(v) = \beta_\pi(v)$, then $\bigl\< \pi_{\downcl\<R,1\>}\bigl(\cB(\beta)(s_1)\bigr), \ldots, \pi_{\downcl\<R,m\>}\bigl(\cB(\beta)(s_m)\bigr), \beta(v) \bigr\> \in R^\cA$ and thus $\cB, \beta \models R(s_1, \ldots, s_m, v)$.
						
						\item If $\beta(v) \neq \beta_\pi(v)$, then $\bigl\< \pi_{\downcl\<R',1\>}\bigl(\cB(\beta)(t_1),\bigr) \ldots, \pi_{\downcl\<R',1\>}\bigl(\cB(\beta)(t_{m'})\bigr), \beta(v) \bigr\> \not\in R'^\cA$ and thus $\cB, \beta \not\models R'(t_1, \ldots, t_{m'}, v)$.
					\end{description}
					In both cases we end up with $\cB,\beta \models R'(t_1, \ldots, t_{m'}, v) \rightarrow R(s_1, \ldots, s_m, v)$.
			\end{description}
			Consequently, we can derive $\cB, \beta \models C$ in all sub-cases.

		\item Case $\cA, \beta_\pi \not\models s\approx t$ for some free atom $s\approx t \in \Gamma$. 
			Since $C$ is in normal form, $s$ and $t$ must be constant symbols.
			$\cB$ and $\cA$ interpret constant symbols in the same way and independently of a variable assignment and thus we immediately get $\cB, \beta\not\models s\approx t$.
		
		\item Case $\cA, \beta_\pi \models s\approx t$ for some $s\approx t \in \Delta$. 
			\begin{description}
				\item If $s$ and $t$ are constant symbols, we know that $\cB, \beta \models s\approx t$ holds, by analogy to the above case.
				
				\item If $s$ is a free-sort variable $v$ and $t$ is a constant symbol $d$, we have $\beta_\pi(v) = d^\cA$. 
					\begin{description}
						\item Suppose $v \approx d$ is guarded by some atom $R(t_1, \ldots, t_m, v)$ in $\Gamma$ with $R$ being marked.
							As done previously, we may assume that $\cA, \beta_\pi \models R(t_1, \ldots, t_m, v)$.					
							Hence, we have\\
							$\beta_\pi(v) = \tau_{R}^\cA\bigl( \cA(\beta_\pi)(t_1), \ldots, \cA(\beta_\pi)(t_m) \bigr)$\\ 
							\strut\hspace{6ex} $= \tau_{R}^\cA\bigl( \pi_{\downcl\<R,1\>}\bigl(\cB(\beta)(t_1)\bigr), \ldots, \pi_{\downcl\<R,m\>}\bigl(\cB(\beta)(t_m)\bigr) \bigr)$.
							\begin{description}
								\item If $\beta(v) = \beta_\pi(v)$, then $\beta(v) = d^\cA = d^\cB$ and thus $\cB, \beta \models v \approx d$.
								
								\item If $\beta(v) \neq \beta_\pi(v)$, then $\bigl\< \pi_{\downcl\<R,1\>}\bigl(\cB(\beta)(t_1)\bigr), \ldots, \pi_{\downcl\<R,m\>}\bigl(\cB(\beta)(t_m)\bigr),$ $\beta(v) \bigr\> \not\in R^\cA$ and thus $\cB, \beta \not\models R(t_1, \ldots, t_m, v)$.
							\end{description}
							In both cases we can derive $\cB,\beta \models R(t_1, \ldots, t_m, v) \rightarrow v \approx d$.
					
						\item Now suppose that $v \approx d$ is not guarded.
							In this case we know that $\hI_{\downcl(v)}$ contains all free-sort constant symbol occurring in $N$ and also all artificial instantiation points $d_{R c_1 \ldots c_m}$.
							Therefore and by the definition of $\cS^\cB$, $\pi_{\downcl(v)}$ can only project $\beta(v)$ to $d^\cA$, if $\beta(v)$ equals $d^\cA$ in the first place.
							Hence, $\beta_\pi(v) = \pi_{\downcl(v)}(\beta(v)) = d^\cA = \beta(v)$.
							This entails $\cB,\beta \models v \approx d$.
					\end{description}	
					
				\item Suppose $s$ is a free-sort variable $v$ and $t$ is a free-sort variable $w$.
					\begin{description}
						\item If there are guards for both variables $v$ and $w$, i.e.\ $\Gamma$ contains two atoms $R(s_1, \ldots, s_m, v)$ and $R'(t_1, \ldots, t_{m'}, w)$ with marked $R$ and $R'$, then we assume $\cA,\beta_\pi \models R(s_1, \ldots,$ $s_m, v)$ and $\cA,\beta_\pi \models R'(t_1, \ldots, t_{m'}, w)$, as in previous cases. Hence,\\
							$\beta_\pi(v) = \tau_R^\cA\bigl( \cA(\beta_\pi)(s_1), \ldots, \cA(\beta_\pi)(s_m) \bigr)$\\
							\strut\hspace{10ex}$= \tau_{R}^\cA\bigl( \pi_{\downcl\<R,1\>}\bigl(\cB(\beta)(s_1)\bigr), \ldots, \pi_{\downcl\<R,m\>}\bigl(\cB(\beta)(s_m)\bigr) \bigr)$\\
							and\\
							$\beta_\pi(w) = \tau_{R'}^\cA\bigl( \cA(\beta_\pi)(t_1), \ldots, \cA(\beta_\pi)(t_m) \bigr)$\\
							\strut\hspace{10ex}$ = \tau_{R'}^\cA\bigl( \pi_{\downcl\<R',1\>}\bigl(\cB(\beta)(t_1)\bigr), \ldots, \pi_{\downcl\<R',m\>}\bigl(\cB(\beta)(t_m)\bigr) \bigr)$\\
							and $\beta_\pi(v) = \beta_\pi(w)$.
							\begin{description}
								\item Suppose $\beta(v) = \beta(w)$. $\cB,\beta \models v \approx w$ follows immediately.
								\item Suppose $\beta(v) \neq \beta(w)$. Hence, we either have \\
									$\beta(v) \neq \tau_{R}^\cA\bigl( \pi_{\downcl\<R,1\>}\bigl(\cB(\beta)(s_1)\bigr), \ldots, \pi_{\downcl\<R,m\>}\bigl(\cB(\beta)(s_m)\bigr) \bigr)$, which entails $\cB,\beta \not \models R(s_1, \ldots, s_m, v)$,
									or\\
									$\beta(w) \neq \tau_{R'}^\cA\bigl( \pi_{\downcl\<R',1\>}\bigl(\cB(\beta)(t_1)\bigr), \ldots, \pi_{\downcl\<R',m\>}\bigl(\cB(\beta)(t_m)\bigr) \bigr)$, which implies $\cB,\beta \not\models R'(t_1, \ldots, t_{m'}, w)$.									
							\end{description}
							In both cases, we have $\cB,\beta \models R(s_1, \ldots, s_m, v) \wedge R'(t_1, \ldots, t_{m'}, w) \rightarrow v \approx w$, and thus also $\cB,\beta \models C$.
					
						\item If at least one of the variables is unguarded, we know that both $\hI_{\downcl(v)}$ and $\hI_{\downcl(w)}$ contain all free-sort constant symbol occurring in $N$ and also all artificial instantiation points $d_{R c_1 \ldots c_m}$. In fact, it even holds $\hI_{\downcl(v)} = \hI_{\downcl(w)}$.
							Analogously to previous cases, we observe $\beta(v) = \beta_\pi(v) = \beta_\pi(w) = \beta(w)$.
							Consequently, we have $\cB,\beta \models v \approx w$.						
					\end{description}		
			\end{description}	
	\end{description}
	Altogether, we have shown $\cB\models N$.	
\end{proof}

\begin{lemma}\label{lemma:ArtificialAndExistentialInstantiationPoints}
	The hierarchic interpretation $\cB$ constructed in the proof of Lemma~\ref{lemma:MarkedPredicatesAsFunctionGraphsInFEU} is a model of $M_n$.
\end{lemma}
	Before we proceed with the proof, we need to update Definition~\ref{definition:FreeInstantiationPoints} (instantiation points for free-sort argument positions) in order to adapt it to the new situation with marked predicate symbols and guarded free-sort atoms.
	To this end, we replace Condition~\ref{enum:FreeInstantiationPoints:II} in Definition~\ref{definition:FreeInstantiationPoints} with the following condition.
	\begin{itemize}
		\item[(b)]  For any clause $\Lambda \,\|\, \Gamma \rightarrow \Delta$ in $N$ such that $\Gamma\rightarrow \Delta$ contains $P(\ldots, u, \ldots)$ in which $u$ occurs as the $i$-th argument and $\Delta$ contains an atom of the form $u \approx t$ where $t$ is either a variable or a constant symbol, we set
			\begin{enumerate}[label=(b.\arabic{*}), ref=(b.\arabic{*})]
				\item $d \in \cI_{P,i}$, if $t$ is some constant symbol $d$ and if there is a guard $R(s_1, \ldots,$ $s_m, u)$ such that $R$ is marked in $N$.
				\item $\cI_{P,i} = \fconsts(N)$, if $u \approx t$ is unguarded.
			\end{enumerate}
	\end{itemize}
\begin{proof}[Proof sketch]
	We already know that $\cB \models N$.
	Hence, in order to proof the lemma, we have to show two things:
	\begin{enumerate}[label=(\arabic{*}), ref=(\arabic{*})]
		\item\label{enum:proofArtificialAndExistentialInstantiationPoints:I} $\hI_{\downcl_N\<P,i\>} = \cI_{\downcl_{M_n}\<P,i\>}$ for every argument position pair and
		\item\label{enum:proofArtificialAndExistentialInstantiationPoints:II} $\cB \models \hPhi(R_1, M_0) \cup \hPhi(R_2, M_1) \cup \ldots \cup \hPhi(R_n,M_{n-1})$.
	\end{enumerate}
	
	\paragraph{Ad~\ref{enum:proofArtificialAndExistentialInstantiationPoints:I}.}
		The requirement~\ref{enum:ArtificialInstantiationPoints:III} regarding the artificial instantiation points in $\hI_{\downcl\<P,i\>}$ does only play a role for argument position pairs $\<P,i\>$ in which either $P$ is unmarked or $i$ is the last argument position in $P$.
		The reason is, on the one hand, that in~\ref{enum:ArtificialInstantiationPoints:III} the existence of an unguarded free-sort atom $u \approx t$ with $\downcl_N(u) = \downcl_N\<P,i\>$ is required.
		On the other hand, Condition~\ref{enum:StratifiedClauseSet:II:II} in Definition~\ref{definition:StratifiedClauseSet} states that $\downcl_N(u) \cap \downcl_N\<R,j\> = \emptyset$ for every marked $R : \xi_1 \times \ldots \times \xi_m \times \xi_{m+1}$ and $j = 1, \ldots, m$.
		This means, any set $\hI_{\downcl_N\<P,i\>}$ that is subject to the requirement~\ref{enum:ArtificialInstantiationPoints:III} cannot participate \emph{as source} in the generation of new artificial instantiation points by means of the requirement~\ref{enum:ArtificialInstantiationPoints:II}.
		However, it could participate \emph{as target} of requirement~\ref{enum:ArtificialInstantiationPoints:II}.
		But this would not lead to new instantiation points, as requirement~\ref{enum:ArtificialInstantiationPoints:III} already covers all possibilities.
		
		Before we continue, we show a technical result.
		\begin{description}
			\item\underline{Claim I:}
				Consider two predicate symbols $R: \xi_1 \times \ldots \times \xi_m \times \xi_{m+1}$ and $R': \zeta_1 \times \ldots \times \zeta_m \times \zeta_{m'+1}$ that are marked in $N$.					
				For every $i$, $1\leq i\leq m'$, $\<R,m+1\> \prop_N \<R',i\>$ entails $R' \not\succeq_N R$.
				
			\item\underline{Proof:}
				Since $N$ is stratified with respect to $R$ and $R'$ and because of\\
				$\<R,m+1\> \prop_N \<R',i\>$, we observe
				$\min_{1 \leq j \leq m+1} \lvl_N\<R,j\> = \lvl_N\<R,m+1\> = \lvl_N\<R',i\> > \lvl_N\<R',m'+1\>$.
				Suppose $R' \succeq_N R$, i.e.\ $\lvl_N\<R',m'+1\> \geq \min_{1 \leq j \leq m+1} \lvl_N\<R,j\>$.
				This contradicts the above observation.
				\strut\hfill$\Diamond$
		\end{description}
		
		Considering the sets of artificial instantiation points, it is clear that any point $d_{R c_1 \ldots c_{m}}$ can only be generated by an application of requirement~\ref{enum:ArtificialInstantiationPoints:II}.
		\begin{description}
			\item\underline{Claim II:}
				For every $d_{R_i c_1 \ldots c_{m_i}}$ that is generated because of requirement~\ref{enum:ArtificialInstantiationPoints:II}, we have \linebreak 
				$d_{R_i c_1 \ldots c_{m_i}} \in \cI_{\downcl_{M_i}\<R_i, m_1+1\>}$.
				
			\item\underline{Proof:}
				We proceed by induction from $R_1$ to $R_n$.				
				\begin{description}
					\item Consider $R_1 : \xi_1 \times \ldots \times \xi_{m_1} \times \xi_{m_1+1}$.
						For $j = 1, \ldots, m_1$ we observe $\hI_{\downcl\<R_1,j\>} = \cI_{\downcl\<R_1,j\>}$, since neither requirement~\ref{enum:ArtificialInstantiationPoints:II} nor~\ref{enum:ArtificialInstantiationPoints:III} introduces artificial instantiation points into $\hI_{\downcl\<R_1,j\>}$.
						Requirement~\ref{enum:ArtificialInstantiationPoints:II} generates exactly the instantiation points in 
							$\bigl\{d_{R_1 c_1 \ldots c_{m_1}} \bigm| c_1 \in \hI_{\downcl_N\<R_1,1\>}, \ldots, c_{m_1} \in \hI_{\downcl_N\<R_1,m_1\>} \bigr\}$
						for $R_1$.
						On the other hand, the definition of $\hPhi(R_1, M_0) = \hPhi(R_1, N)$ leads to
							$\bigl\{R_1( c_1, \ldots, c_{m_1}, d_{R_1 c_1 \ldots c_{m_1}}) \bigm| c_1 \in \cI_{\downcl_N\<R_1,1\>}, \ldots, c_{m_1} \in \cI_{\downcl_N\<R_1,m_1\>} \bigr\} \subseteq \hPhi(R_1,N)$.
						Hence, $\bigl\{d_{R_1 c_1 \ldots c_{m_1}} \bigm| c_1 \in \hI_{\downcl_N\<R_1,1\>}, \ldots, c_{m_1} \in  \hI_{\downcl_N\<R_1,m_1\>} \bigr\}$\\
						$ = \bigl\{d_{R_1 c_1 \ldots c_{m_1}} \bigm| c_1 \in \cI_{\downcl_N\<R_1,1\>}, \ldots, c_{m_1} \in \cI_{\downcl_N\<R_1,m_1\>} \bigr\} \subseteq \cI_{\downcl_{M_1}\<R_1,m_1+1\>}$.
					
					\item Consider $R_\ell : \xi_1 \times \ldots \times \xi_{m_\ell} \times \xi_{m_\ell+1}$ with $\ell > 1$.
						Moreover, consider any $\<R_\ell,j\>$ with $j \leq m_\ell$.
						We have pointed out earlier, that none of the artificial instantiation points in $\hI_{\downcl_N\<R_\ell,j\>} \setminus \cI_{\downcl_N\<R_\ell,j\>}$ with $j = 1, \ldots, m_j$ belongs to $\hI_{\downcl_N\<R_\ell,j\>}$ because of requirement~\mbox{\ref{enum:ArtificialInstantiationPoints:III}}.
						For any $d_{R_k c_1 \ldots c_{m_k}} \in \hI_{\downcl_N\<R_\ell,j\>} \setminus \cI_{\downcl_N\<R_\ell,j\>}$ we must have $\<R_k, m_k+1\> \prop_N \<R_\ell,j\>$, which, by Claim I and our assumption $R_1 \succeq_N \ldots \succeq_N R_n$, entails $k < \ell$.
						By induction, we have $\hI_{\downcl_{N}\<R_k, m_k+1\>} \subseteq \cI_{\downcl_{M_k}\<R_k, m_k+1\>} \subseteq \cI_{\downcl_{M_k}\<R_\ell,j\>} \subseteq \cI_{\downcl_{M_{\ell-1}}\<R_\ell,j\>}$.
						Consequently, we have $\hI_{\downcl_{n}\<R_\ell,j\>} \subseteq \cI_{\downcl_{M_{\ell-1}}\<R_\ell,j\>}$.

						Since requirement~\ref{enum:ArtificialInstantiationPoints:II} generates exactly the instantiation points in \\ 
							$\bigl\{d_{R_\ell c_1 \ldots c_{m_\ell}} \bigm| c_1 \in \hI_{\downcl_N\<R_\ell,1\>}, \ldots, c_{m_\ell} \in \hI_{\downcl_N\<R_\ell,m_\ell\>} \bigr\}$
						for $\<R_\ell, m_\ell+1\>$ and due to \\
							$\bigl\{R_\ell( c_1, \ldots, c_{m_\ell}, d_{R_\ell c_1 \ldots c_{m_\ell}}) \bigm| c_1 \in \cI_{\downcl_N\<R_\ell,1\>}, \ldots, c_{m_\ell} \in \cI_{\downcl_N\<R_\ell,m_\ell\>} \bigr\} \subseteq \hPhi(R_\ell,M_{\ell-1})$,
						we obtain \\
						$\bigl\{d_{R_\ell c_1 \ldots c_{m_\ell}} \bigm| c_1 \in \hI_{\downcl_N\<R_\ell,1\>}, \ldots, c_{m_\ell} \in \hI_{\downcl_N\<R_\ell,m_\ell\>} \bigr\}$\\
						$\subseteq \bigl\{d_{R_\ell c_1 \ldots c_{m_\ell}} \bigm| c_1 \in \cI_{\downcl_{M_{\ell-1}}\<R_\ell,1\>}, \ldots, c_{m_\ell} \in \cI_{\downcl_{M_{\ell-1}}\<R_\ell,m_\ell\>} \bigr\} \subseteq \cI_{\downcl_{M_\ell}\<R_\ell,m_\ell+1\>}$.
						\strut\hfill$\Diamond$
				\end{description}	
		\end{description}
			
		\begin{description}
			\item\underline{Claim III:}
				Let $R_i : = \xi_1 \times \ldots \times \xi_{m_i} \times \xi_{m_i+1}$ be marked in $N$.
				For every $d_{R_i c_1 \ldots c_{m_i}} \in \cI_{\downcl_{M_n}\<Q,j\>}$ with $\<R_i, m_i+1\> \prop_{M_n} \<Q,j\>$ we have $d_{R_i c_1 \ldots c_{m_i}} \in \hI_{\downcl_{N}\<Q,j\>}$.
				
			\item\underline{Proof:}
				We proceed by induction from $R_1$ to $R_n$.
				\begin{description}
					\item Consider $R_1 : \xi_1 \times \ldots \times \xi_m \times \xi_{m_1+1}$.
						Whenever $d_{R_1 c_1 \ldots c_{m_1}}$ belongs to $\cI_{\downcl_{M_n}\<P,i\>}$, then the atom $R_1(c_1, \ldots, c_{m_1}, d_{R_1 c_1 \ldots c_{m_1}})$ must occur in $\hPhi(R_1,$ $N)$.
						Hence, we have $\<c_1, \ldots, c_{m_1}\> \in \cI_{\downcl_{N}\<R_1,1\>} \times \ldots \times \cI_{\downcl_{N}\<R_1,m_1\>}$.
						Because of requirement~\ref{enum:ArtificialInstantiationPoints:I} regarding artificial instantiation points, $\<c_1, \ldots, c_{m_1}\>$ must also belong to $\hI_{\downcl_{N}\<R_1,1\>} \times \ldots \times \hI_{\downcl_{N}\<R_1,m_1\>}$.
						
						Our assumption $\<R_1, m_1+1\> \prop_{M_n} \<Q,j\>$ can only be satisfied if $\<R_1, m_1+1\> \prop_{N} \<Q,j\>$ holds.
						This, in turn, entails $R_1 \succeq_N Q$.
						Taken together, requirement~\ref{enum:ArtificialInstantiationPoints:II} leads to $d_{R_1 c_1 \dots c_{m_1}} \in \hI_{\downcl_N\<Q,j\>}$.
				
					\item Consider $R_\ell : \xi_1 \times \ldots \times \xi_m \times \xi_{m_\ell+1}$ with $\ell > 1$.
						Whenever $d_{R_\ell c_1 \ldots c_{m_\ell}}$ belongs to $\cI_{\downcl_{M_n}\<P,i\>}$, then the atom $R_\ell(c_1, \ldots, c_{m_\ell}, d_{R_\ell c_1 \ldots c_{m_\ell}})$ must occur in $\hPhi(R_\ell,M_{\ell-1})$.
						Hence, we have $\<c_1, \ldots, c_{m_\ell}\> \in \cI_{\downcl_{M_{\ell-1}}\<R_\ell,1\>} \times \ldots \times \cI_{\downcl_{M_{\ell-1}}\<R_\ell,m_\ell\>}$.
						Every instantiation point in any set $\cI_{\downcl_{M_{\ell-1}}\<R_\ell, k\>} \setminus \cI_{\downcl_{N}\<R_\ell, k\>}$, $1 \leq k \leq m_\ell$, has been propogated into the set $\cI_{\downcl_{M_{\ell-1}}\<R_\ell, k\>}$ via $\prop_N$, because our syntax does not allow any unguarded free-sort atom $u \approx t$ with $\downcl_N(u) = \downcl_N\<R_\ell,k\>$.						
						Thus, induction entails $\<c_1, \ldots, c_{m_\ell}\> \in \hI_{\downcl_{N}\<R_\ell,1\>} \times \ldots \times \hI_{\downcl_{N}\<R_\ell,m_\ell\>}$.
						
						Our assumption $\<R_\ell, m_\ell+1\> \prop_{M_n} \<Q,j\>$ can only hold if $\<R_\ell, m_\ell+1\> \prop_{N} \<Q,j\>$ holds.
						This, in turn, entails $R_\ell \succeq_N Q$.
						Taken together, requirement~\ref{enum:ArtificialInstantiationPoints:II} leads to $d_{R_\ell c_1 \dots c_{m_\ell}} \in \hI_{\downcl_N\<Q,j\>}$.
						\strut\hfill$\Diamond$
				\end{description}	
		\end{description}
										
		Let $\<P,i\>$ be an argument position pair such that there is an unguarded free-sort atom $u \approx t$ in $N$ for which $\downcl_N(u) = \downcl_N\<P,i\>$.		
		Due to Claim~II, we have $\hI_{\downcl_N\<Q,j\>} \subseteq \cI_{\downcl_{M_n}\<Q,j\>}$ for every argument position pair $\<Q,j\>$.
		Hence, $\hI_{\downcl_N\<P,i\>} = \bigcup_{\<Q,j\>} \hI_{\downcl_N\<Q,j\>} \subseteq \cI_{\downcl_{M_n}\<P,i\>}$.
		
		Conversely, we have $\cI_{\downcl_{M_n}\<P,i\>} = \fconsts(M_n)$ and we can split $\cI_{\downcl_{M_n}\<P,i\>}$ into $\cI_{\downcl_N\<P,i\>}$ and the rest $\cI_{\downcl_{M_n}\<P,i\>}\setminus \cI_{\downcl_N\<P,i\>}$.
		Every instantiation point in this rest is of the form $d_{R c_1 \ldots c_m}$ and it belongs to $\cI_{\downcl_{M_n}\<R,m+1\>}$.
		In addition, we observe $\<R,m+1\> \prop_{M_n} \<R,m+1\>$.
		Hence, Claim III implies that $\cI_{\downcl_{M_n}}\<R,m+1\> \subseteq \hI_{\downcl_N}\<R,m+1\>$.
		Moreover, by requirement~\ref{enum:ArtificialInstantiationPoints:I}, we know $\cI_{\downcl_N\<P,i\>} \subseteq \hI_{\downcl_N\<P,i\>}$.
		Taken together, this entails $\cI_{\downcl_{M_n}\<P,i\>} \subseteq \hI_{\downcl_N\<P,i\>}$.
		
		\medskip
		Consequently, for every (arbitrary) argument position pair $\<P,i\>$ in $N$, we may conclude $\cI_{\downcl_{M_n}\<P,i\>} = \hI_{\downcl_N\<P,i\>}$ by Claim~II, Claim~III, the just made observations concerning the unguarded free-sort atoms $u \approx t$, and the requirement~\ref{enum:ArtificialInstantiationPoints:I} stating $\cI_{\downcl_N\<P,i\>} \subseteq \hI_{\downcl_N\<P,i\>}$.
																								
	\paragraph{Ad~\ref{enum:proofArtificialAndExistentialInstantiationPoints:II}.}
		
		Let $R_\ell(c_1, \ldots, c_{m_\ell}, d_{R_\ell c_1 \ldots c_{m_\ell}}) \in \hPhi(R_\ell,M_{\ell-1})$ for some $\ell$, $1\leq \ell\leq n$.
		By construction of $\cB$, we known that $d_{R_\ell c_1 \ldots c_{m_\ell}}^\cB = d_{R_\ell c_1 \ldots c_{m_\ell}}^\cA = \tau_{R_\ell}^\cA(c_1^\cA, \ldots, c_{m_\ell}^\cA) = \tau_{R_\ell}^\cA(c_1^\cB, \ldots, c_{m_\ell}^\cB)$.
		Hence, $\cB \models R_\ell(c_1, \ldots, c_{m_\ell}, d_{R_\ell c_1 \ldots c_{m_\ell}})$.
		
		More generally, for every $\<\fa_1, \ldots, \fa_{m_\ell}, \fb\>$ of domain elements we have $\<\fa_1, \ldots, \fa_{m_\ell}, \fb\> \in R_\ell^\cB$ if and only if
			$\fb = \tau_{R_\ell}^\cA\bigl( \pi_{\downcl_{N}\<R_\ell,m_1\>}(\fa_1), \ldots, \pi_{\downcl_{N}\<R_\ell,m_\ell\>}(\fa_{m_\ell}) \bigr)$.
		By definition of the projections $\pi_{\downcl_N\<R_\ell,i\>}$, there must be a tuple 
		$\<c_1, \ldots, c_{m_\ell}\> 
			\in \hI_{\downcl_N\<R_\ell,1\>} \times \ldots \times \hI_{\downcl_N\<R_\ell,m_\ell\>} 
			= \cI_{\downcl_{M_n}\<R_\ell,1\>} \times \ldots \times \cI_{\downcl_{M_n}\<R_\ell,m_\ell\>}
			= \cI_{\downcl_{M_\ell}\<R_\ell,1\>} \times \ldots \times \cI_{\downcl_{M_\ell}\<R_\ell,m_\ell\>}$
		(the last equation is valid, as the $\hI_{\downcl_N\<R_\ell,i\>}$ with $i \leq m_\ell$ are not affected by \ref{enum:ArtificialInstantiationPoints:III}), 
		such that 
		$\<\pi_{\downcl_{N}\<R_\ell,m_1\>}(\fa_1), \ldots, \pi_{\downcl_{N}\<R_\ell,m_\ell\>}(\fa_{m_\ell})\> 
			= \<c_1^\cB, \ldots, c_{m_\ell}^\cB\>$ 
		and, hence, 
		$\fb = \tau_{R_\ell}^\cA(c_1^\cB, \ldots, c_{m_\ell}^\cB) 
			= d_{R_\ell c_1 \ldots c_{m_\ell}}^\cB$.
		
		From this observation $\cB \models \hPhi(R_\ell,M_{\ell-1})$ follows.
\end{proof}

\end{document}